\documentclass[a4paper,UKenglish,cleveref, autoref, thm-restate]{lipics-v2021}

\usepackage{amsmath, amssymb}
\usepackage{mathtools}
\usepackage{dashbox}
\usepackage{bm}
\usepackage{hyperref}
\usepackage{paralist}
\usepackage{graphicx}
\usepackage{cite}
\usepackage{amssymb}
\usepackage{setspace}
\usepackage{colortbl}
\usepackage[table]{xcolor}
\usepackage{textcomp}
\usepackage{tabularx}
\usepackage{adjustbox}
\usepackage{bm}
\usepackage{thmtools} 
\usepackage{dsfont}
\usepackage{wrapfig}
\usepackage{tcolorbox}
\usepackage{wasysym}
\usepackage{thmtools}
\usepackage{tikz}
\usetikzlibrary{fit}
\usetikzlibrary{shapes}
\usetikzlibrary{decorations.markings}
\usetikzlibrary{arrows}
\usetikzlibrary{shapes.geometric}
\usepackage{tkz-base,tkz-fct}
\usepackage{pgfplots}
\usepackage{pgfplotstable}
\pgfplotsset{compat=1.18}

\usepackage{pifont}
\usepackage{enumitem}
\usepackage[strings]{underscore}
\usepackage{csquotes} 

\newcommand{\etal}

\newcommand{\Z}{\ensuremath{\mathbb{Z}}}

\def\top{\mathrm{top}}
\def\bottom{\mathrm{bottom}}
\def\H{\mathcal{H}}

\hideLIPIcs
\nolinenumbers

\title{Lozenge Tiling  by Computing Distances}

\newlength{\strutdepth}%
\newcommand{\mycomment}[3]{%
    \noindent{\bfseries
    \color{#2}{#1}\color{black}}%
    \strut\vadjust{\kern-\strutdepth%
        \vtop to \strutdepth{%
            \baselineskip\strutdepth%
            \vss\llap{{\large\color{#2}#3\quad\color{black}}}\null%
        }%
    }%
}

\author{Favreau Jean-Marie}{Université Clermont Auvergne, LIMOS, France \and \url{https://jmfavreau.info/} }{j-marie.favreau@uca.fr}{https://orcid.org/0000-0002-2460-6336}{}

\author{Gerard Yan\footnote{Corresponding author}}{Université Clermont Auvergne, LIMOS, France \and \url{https://yangerard.wordpress.com/} }{yan.gerard@uca.fr}{https://orcid.org/0000-0002-2664-0650}{}

\author{Lafourcade Pascal}{Université Clermont Auvergne, LIMOS, France \and \url{https://sancy.iut.uca.fr/~lafourcade/index.html}}{pascal.lafourcade@uca.fr}{[https://orcid.org/0000-0002-4459-511X]}{ANR PRC grant MobiS5 (ANR-18-CE39-0019), SEVERITAS (ANR-20-CE39-0005), ANR Project PRIVA-SIQ and by the French government IDEX-ISITE initiative 16-IDEX-0001 (CAP 20-25)}

\author{Robert Léo}{Université de Picardie Jules Verne, MIS, France}{leo.robert@u-picardie.fr}{[https://orcid.org/0000-0002-9638-3143]}{ANR PRC grant MobiS5 (ANR-18-CE39-0019)}

\authorrunning{J.M. Favreau, Y. Gerard, P. Lafourcade, L. Robert. } 



\category{} 

\EventEditors{John Q. Open and Joan R. Access}
\EventNoEds{2}
\EventLongTitle{42nd Conference on Very Important Topics (CVIT 2016)}
\EventShortTitle{CVIT 2016}
\EventAcronym{CVIT}
\EventYear{2016}
\EventDate{December 24--27, 2016}
\EventLocation{Little Whinging, United Kingdom}
\EventLogo{}
\SeriesVolume{42}
\ArticleNo{23}

\begin{document}
\maketitle
\begin{abstract}
The Calisson puzzle is a recent tiling game in which one must tile a triangular grid inside a hexagon with lozenges, under the constraint that certain prescribed edges must remain tile boundaries and that adjacent lozenges along these edges have different orientations.
We present the first polynomial-time algorithm for this problem, with running time 
$O(n^3)$ for a hexagon of side length $n$. This algorithm, called the advancing surface algorithm, can be executed in a simple and intuitive way, even by hand with a pencil and an eraser. Its apparent simplicity conceals a deeper algorithmic reinterpretation of the classical ideas of John Conway and William Thurston, which we revisit from a theoretical computer science perspective.

We introduce a graph-theoretic and difference constraints overlay that complements Thurston’s theory of lozenge tilings, revealing its intrinsic algorithmic structure and extending its scope to tiling problems with interior constraints and without necessarily boundary conditions.
In Thurston’s approach, lozenge tilings are lifted to monotone stepped surfaces in the three-dimensional cubic lattice and projected back to the plane using height functions, reducing the tiling problem  to the computation of heights. We show that, at an algorithmic level, selecting a monotone surface corresponds to selecting a directed cut (dicut) in a periodic directed graph, while height functions  are solutions of a system of difference constraints. In this formulation, a region is tilable if and only if the associated weighted directed graph contains no cycle of strictly negative total weight.
This new graph layer completing Thurston's theory shows that Bellman–Ford’s shortest path algorithm is the only algorithmic primitive needed to decide feasibility and compute solutions. In particular, our framework allows us to decide whether the infinite triangular grid can be tiled while respecting a finite set of prescribed local constraints, a setting in which no boundary conditions are available.
\end{abstract}

\section{Introduction}

Lozenge tilings are found in art, architecture, and monuments around the world. They are universal patterns. In science, their geometric and combinatorial properties have intrigued mathematicians for centuries, drawing the attention of renowned researchers like John Conway and William Thurston in recent decades.

In 2022, Olivier Longuet, a mathematics teacher in a French high school, introduced a geometric logic game called \textit{the Calisson puzzle} (original name, \textit{le jeu du calisson}\footnote{The name \textit{Calisson} comes from a traditional French sweet shaped like a lozenge and produced in Aix-en-Provence, a town in the south of France.}). The puzzle consists in tiling a hexagonal region of the triangular grid with lozenges subject to local constraints. 
Olivier Longuet writes a blog (in French) at  \url{https://mathix.org/calisson/blog/} where he presents the rules of the puzzle and more than five hundred instances of the game \cite{blog}.
The game can also be played online at \url{https://martialtarizzo.github.io/Calisson-Game/index.en.html} \cite{online}.

The game is played in a triangular grid bounded by a regular hexagon. The triangular grid contains three types of edges, depending on their orientation: vertical (the $6$/$12$ o'clock direction), the $4$/$10$ o'clock direction, and the $2$/$8$ o'clock direction.
 A lozenge is the union of two adjacent unit triangles. 
 Its type is determined by the orientation of the edge shared by the two triangles.
 There are three types of lozenges, depending on the orientation of their common edge. Lozenges whose common edge is vertical, in the $4$/$10$ o'clock direction, or in the $2$/$8$ o'clock direction are, for instance, colored respectively in yellow, red, and blue.
The rules of the game are simple. They are illustrated in Fig.~\ref{regles}.

\begin{figure}[bth]
  \begin{center}
		\includegraphics[width=0.8\textwidth]{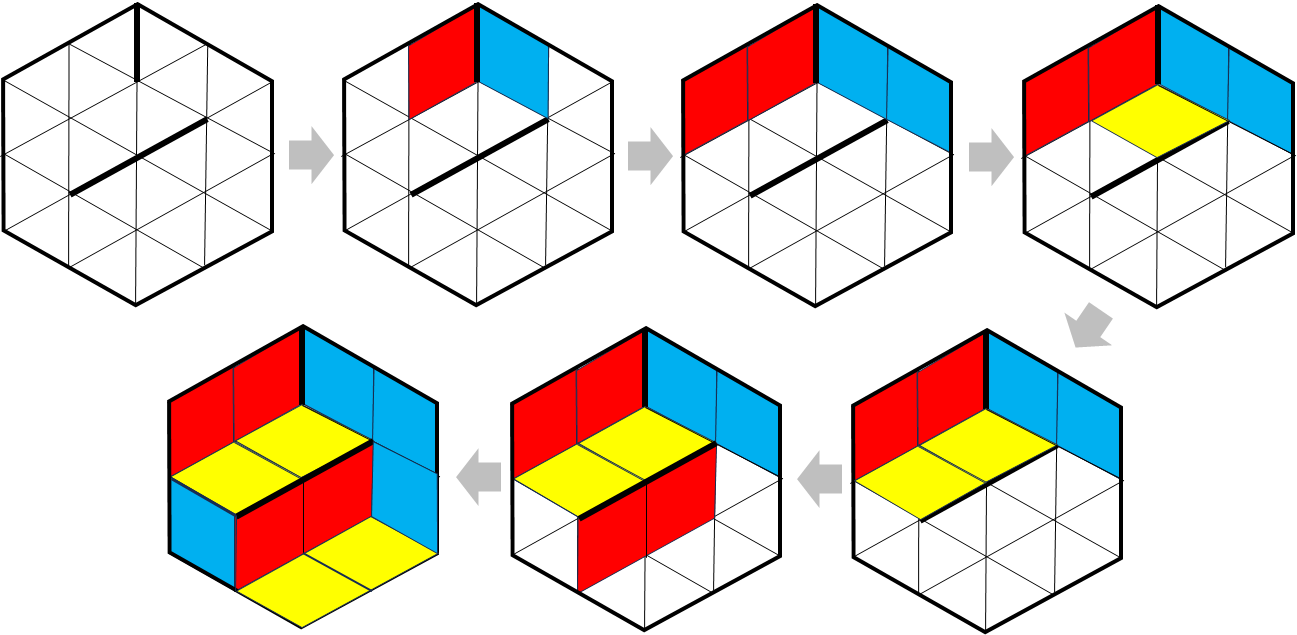}
	\end{center}
	\caption{\label{regles} \textbf{The rules of the puzzle} (Courtesy of Olivier Longuet's \href{https://mathix.org/calisson/blog/}{blog} \cite{blog}): we give ourselves a set of edges, as drawn in in the top left-hand corner. The goal is to tile the hexagon with lozenges in such a way that the edges given as input are adjacent to two lozenges of different colors.}
\end{figure}

\begin{tcolorbox}[colback=red!5!white,
                  colframe=black!75!black,
                  title=Rules of the Calisson puzzle
                 ]

\textbf{\texttt{Input:} } A triangular grid bounded by a regular hexagon, together with a set of edges of the grid denoted by $X_2$ (notation used in the later).

\tcblower
\textbf{\texttt{Goal:} }the problem is to tile the grid with lozenges in such a way that the input edges are not overlapped (such a condition is called \textit{non overlapping constraint}) and are adjacent to two lozenges of different colors (\textit{saliency constraint}).
\end{tcolorbox}

\begin{figure}[ht!]
  \begin{center}
		\includegraphics[width=\textwidth]{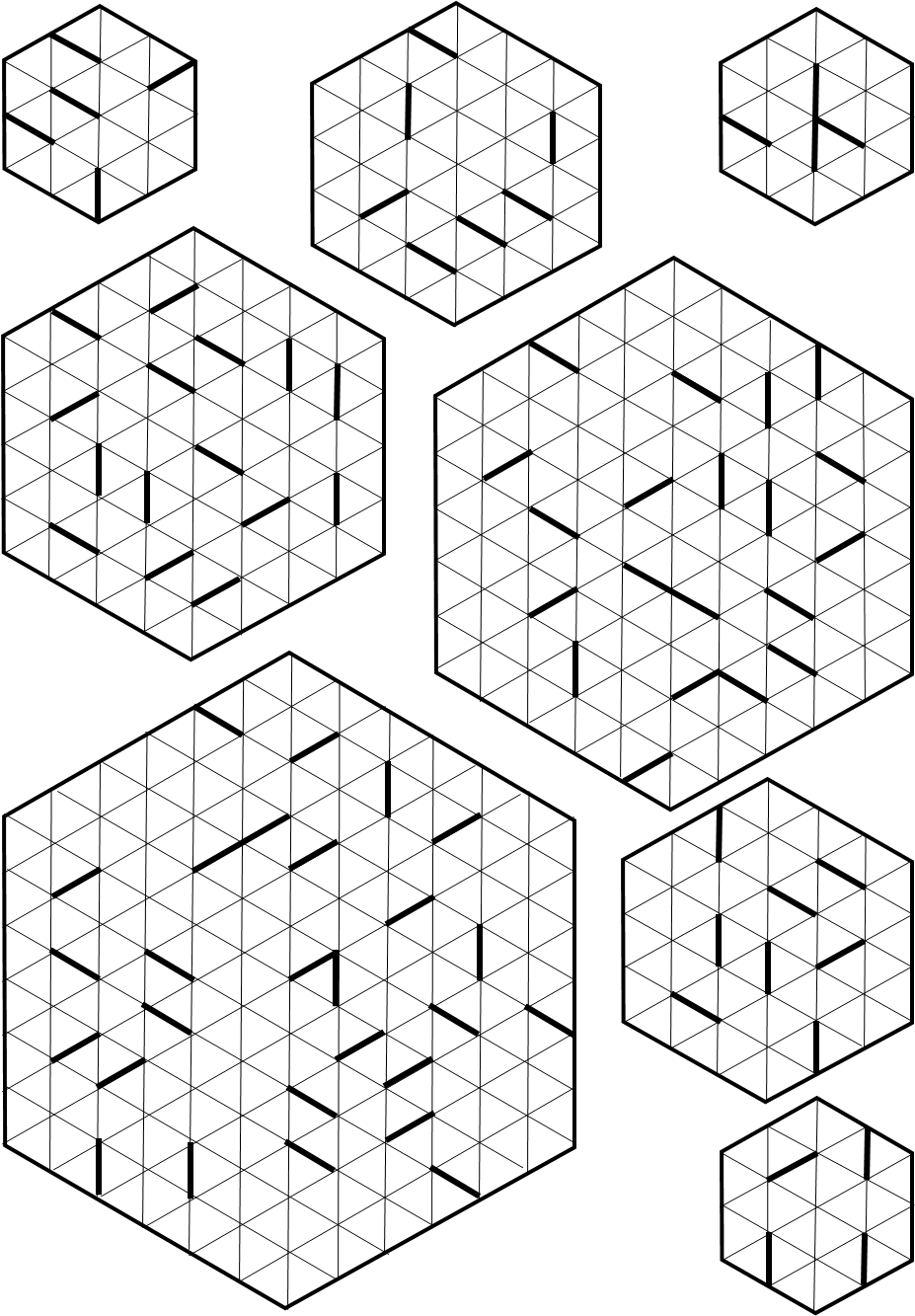}
	\end{center}
	\caption{\label{exercise} \textbf{Instances of the Calisson puzzle.} The instance of size $n=6$ is solved in Fig.~\ref{ads}.}
\end{figure}

The puzzle appeals to the classical observation that lozenge tilings can be viewed as perspective images of stepped surfaces.
As exercise, we invite the reader to solve some of the instances of the puzzle drawn in Fig.~\ref{exercise}.

Surprisingly, the Calisson puzzle introduces a \emph{saliency constraint} (the requirement that lozenges adjacent to a given edge have different orientations) that usual tiling algorithms cannot handle directly. 
The most classical algorithmic strategy for tiling regions by dominoes or lozenges is to reduce the problem to the computation of a matching. The interest of the Calisson puzzle comes from the fact that this strategy fails, as illustrated in Fig.~\ref{mat}, while the other classical algorithm, Thurston's algorithm cannot take into account prescribed interior edges. 
It makes from the computational complexity of the puzzle  an interesting question.
The first contribution of this paper is a polynomial-time algorithm for solving the Calisson puzzle. We called this algorithm the \textit{advancing surface algorithm} since its strategy works by adding cubes so that their surface seems to advance from the back of the hexagon to its front.
This algorithm is sufficiently simple to be carried out with pencil and paper, and is illustrated in Fig.~\ref{start2}.

\begin{figure}[ht]
\begin{center}	\includegraphics[width=0.85\textwidth]{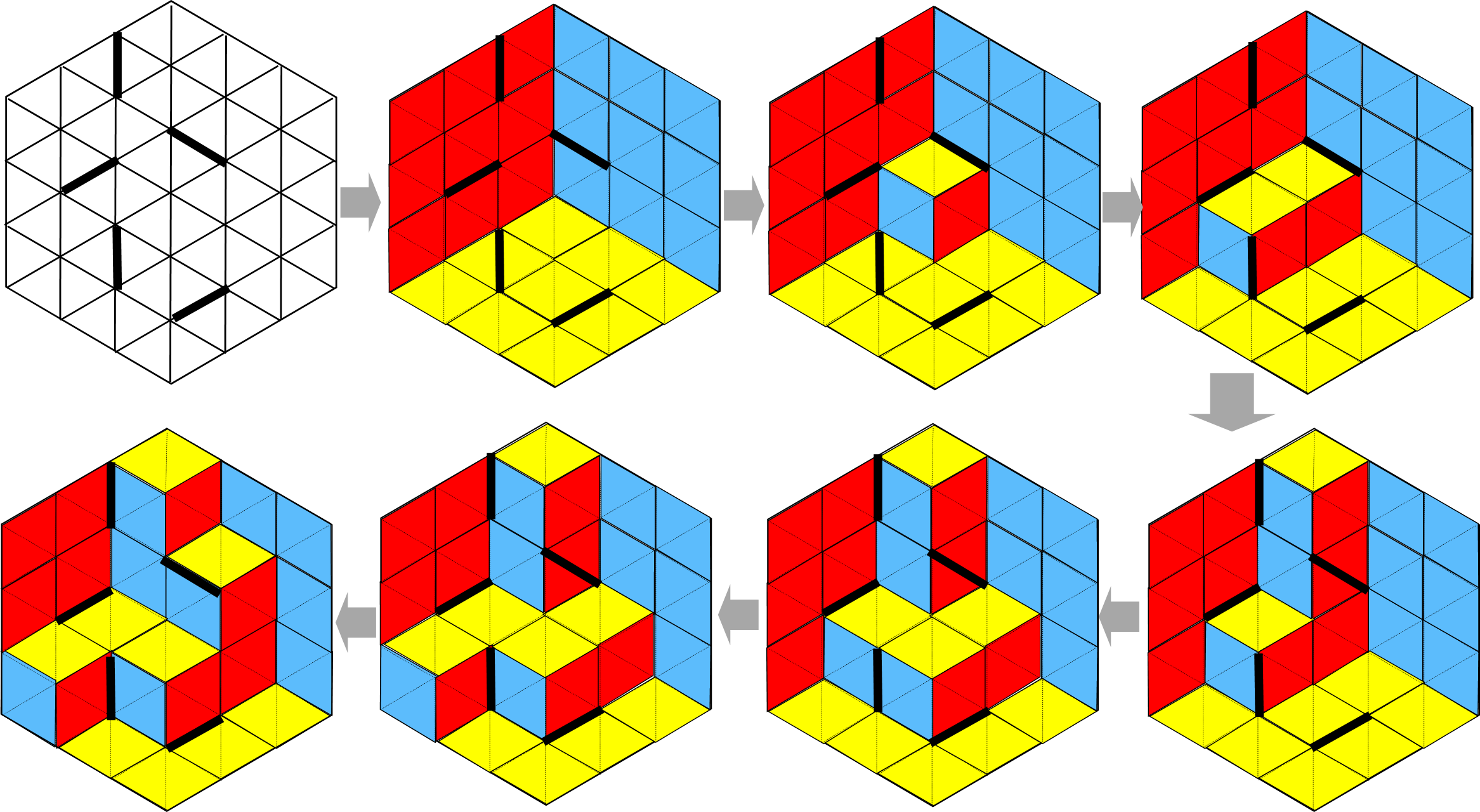}
	\caption{\label{start2} \textbf{The advancing surface  algorithm} for solving an instance of the Calisson puzzle. We start from a tiling looking as the surface of an empty cube and make it advance little by little by adding small cubes in order to satisfy new constraints. It leads sometimes to loose a previously satisfied constraint but it is part of the algorithm.}
 \end{center}
\end{figure}

Since this algorithm is not a straightforward application of existing results, it naturally raises several questions.
Does it extend to arbitrary regions?
How does it relate to the classical theory of lozenge tilings?
Addressing these questions leads us to the central part of the paper where we revisit and complete  the theory and folklore of lozenge tilings from an algorithmic point of view. In Conway and Thurston approach, a lozenge tiling of a simply connected region is lifted to a monotone stepped surface in a three-dimensional cubic lattice.
This surface can then be projected back to the plane using a height function defined on the vertices of the triangular grid.
In this framework, tiling a region reduces to finding a height function satisfying local constraints.

We reformulate Thurston’s theory in the language of directed graphs.
Monotone stepped surfaces in the three-dimensional cubic lattice are interpreted as directed cuts (dicuts) in a periodic directed graph whose vertices correspond to unit cubes.
Within this framework, height functions naturally arise as solutions of systems of difference constraints.
It shows that lozenge tilings can be computed by solving shortest-path problems in weighted directed graphs. With negative weights, the shortest paths algorithm is  Bellman–Ford  \cite{Bellman}. This original process is illustrated in Fig.~\ref{end}.

\begin{figure}[ht]
\begin{center}
\includegraphics[width=0.90\textwidth]{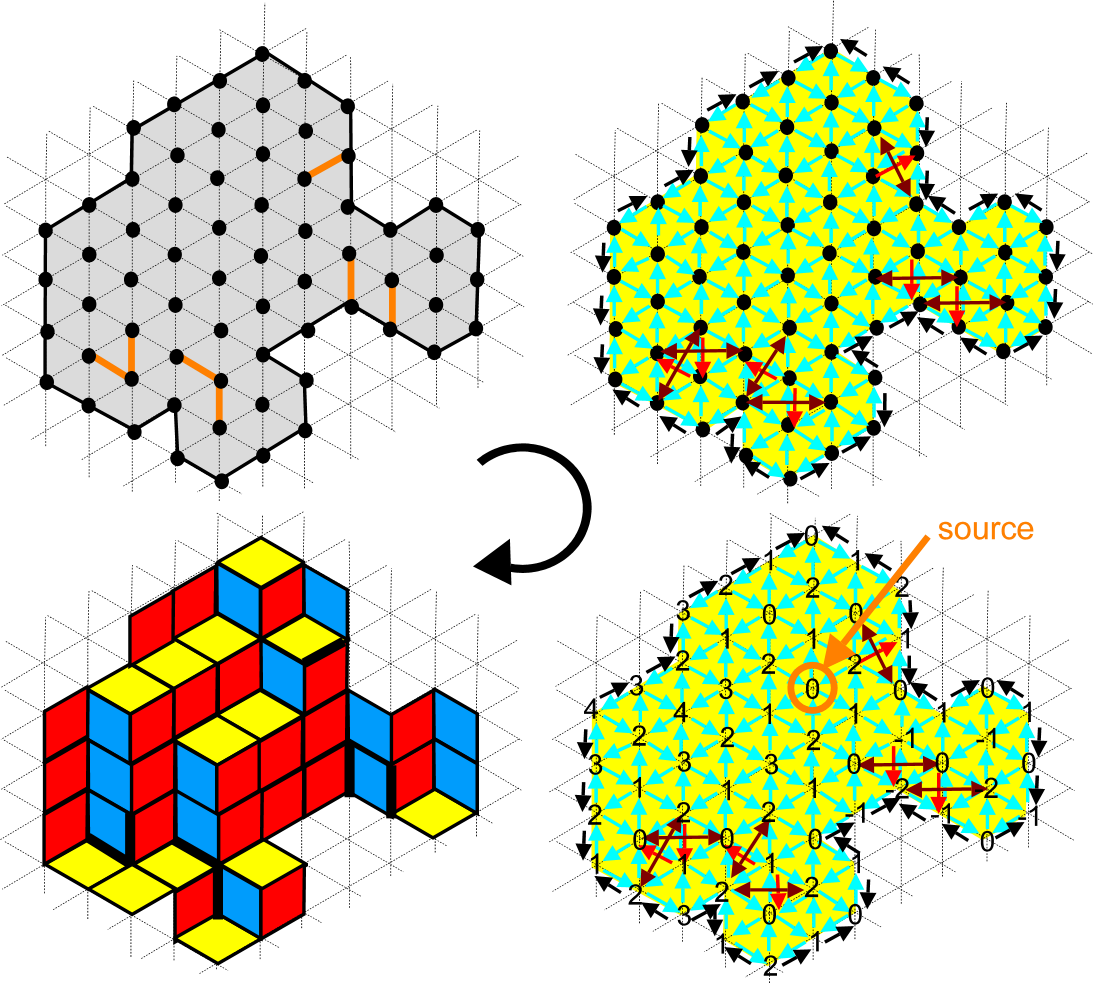}
\end{center}
\caption{\label{end} \textbf{Solving a lozenge tiling instance with non-overlapping and saliency constraints through distance computation.}
The first step builds a weighted directed graph, with weights $+1$ for blue edges, $0$ for brown edges, and $-1$ for black and red edges.
The second step computes shortest-path distances from an arbitrary source vertex $s$.
If the graph contains a cycle of strictly negative total weight, the distance constraints are infeasible and the tiling instance admits no solution.
Otherwise, as shown here, a tiling is recovered by connecting adjacent vertices whose distance to $s$ differ by exactly $1$.}
\end{figure}

The directed graph overlay introduced in this work constitutes a substantial complement of Thurston’s classical theory of lozenge tilings.
It provides a unified and flexible algorithmic framework, allowing one to incorporate interior constraints—such as non-overlapping and saliency constraints—in a simple and systematic way.
To the best of our knowledge, such constraints had not previously been integrated explicitly into the height-function framework.
This approach makes it possible to solve a wide variety of lozenge tiling problems, several of which are illustrated in Fig.~\ref{questions}.
Finally, beyond its theoretical interest, the method has an appealing pedagogical aspect: before understanding the underlying theory, computing a tiling by merely running a shortest-path algorithm, as in Fig.~\ref{end}, has the flavor of a magic trick, which can surprise and attract the attention of young audiences who are sensitive to recreational mathematics.

\begin{figure}[ht]
  \begin{center}
		\includegraphics[width=0.85\textwidth]{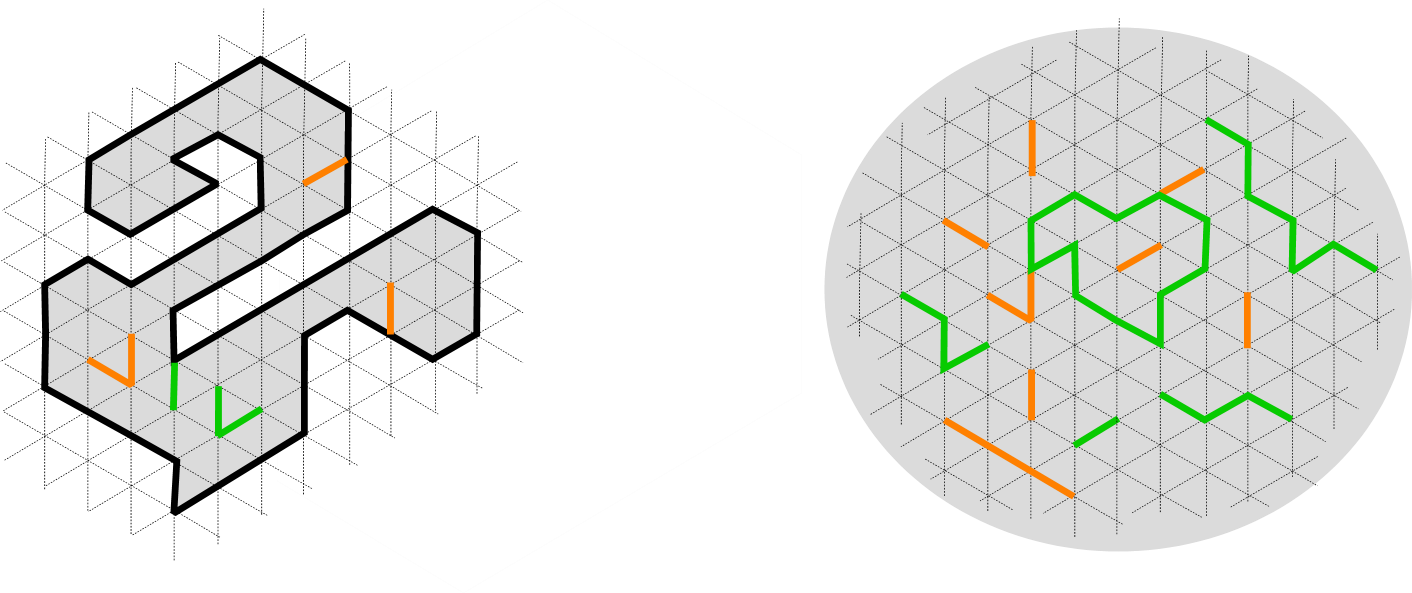}
	\caption{\label{questions} \textbf{The tiling problems that we solve.} The left image is an instance of the problem that we denote \texttt{Tiling$(R,X_1,X_2)$}. Given the finite simply connected  region $R$ and two sets of edges $X_1$ (the green edges) and $X_2$ (the orange edges), the problem is to tile the region $R$ without overlapping the green and orange edges (non overlapping constraint). The lozenges adjacent to orange edges also must have different orientation (saliency constraint). The right image illustrates the toy problem \texttt{Tiling$(\triangle ,X_1,X_2)$} where the region $R$  to be tiled is the whole triangular grid denoted $\triangle $ and thus has no boundary. The main contribution of the paper shows how to solve these problems through a system of difference constraints and thus with distance computations.  }
    \end{center}
\end{figure}

The paper is organized as follows. In Section~2, we present a state of the art of the classical algorithm that can be used for solving the Calisson puzzle and which turn out to fail. In Section~3, we state formaly a generic problem of tilability and present our results.
In Section~4, we explain how the lozenges tiling problems can be reduced to systems of difference constraints. This equivalence passes through directed cuts in a periodic directed graph. The algorithms are presented in details in  Section~5.

\section{Classical Algorithms Fail to Solve the Calisson puzzle}\label{fail}

A reasonable idea for solving Calisson puzzles is to use classical techniques from tiling problems.  The classical Thurston's algorithm for determining whether a region $R$ is tilable by calissons can take into account neither the interior edges of the instance, nor the saliency constraints. It is therefore not directly able to solve the Calisson puzzles without the kind of complement that we describe in the paper. 
However, there are other approaches, either used for tilability by dominos
or for general combinatorial problems. Three methods are worth examining: $3$-SAT, matching in a bipartite graph and a reduction to Maximum Independent Set.\\

{\bf 3-SAT.} The Calisson puzzle is easily expressed as a 3-SAT formula.
Consider a variable $a_c$ for each lozenge $c$ of the region to be tiled. The variable $a_c$ is equal to $1$ if the lozenge $c$ is included in the solution's tiling and $0$ otherwise.
We have four classes of clauses.
\begin{compactenum}
    \item $3$-clauses for expressing the condition that all the triangles of the region must be covered by at least one lozenge 
    (for boundary triangles, these are 2-clauses or even 1-clauses).
    \item The second class of clauses expresses the constraint that the tiles must not overlap. 
    $2$-clauses for avoiding to cover twice a given triangle. We impose $\overline a_c \vee \overline a_{c'}$ to ensure that they do not overlap.
    \item $1$-clauses for expressing the non overlapping constraints of the input edges.
    \item $2$-clauses for expressing the saliency constraints. 
\end{compactenum}

The number of variables and clauses is $O(n^2)$. 
Unfortunatly, this reduction to 3-SAT does not provide a Horn formula and thus it is not known for being solvable in polynomial time.\\

{\bf Matching.}  A classic, non-exponential approach to compute tilings by dominoes (lozenges are unions of two adjacent triangles) is to reduce the problem to the computation of a perfect matching in the graph of adjacency of the triangles (see for example \cite{remila}). This approach is illustrated in Fig.~\ref{mat}. It allows us to tile the region $R$ by taking into account the non overlapping constraints given by the edges of $X$, but not the saliency constraints. Adapting the matching strategy to take into account the saliency constraints does not seem easy.\\ 

\begin{figure}[ht]
  \begin{center}
	\includegraphics[width=0.95\textwidth]{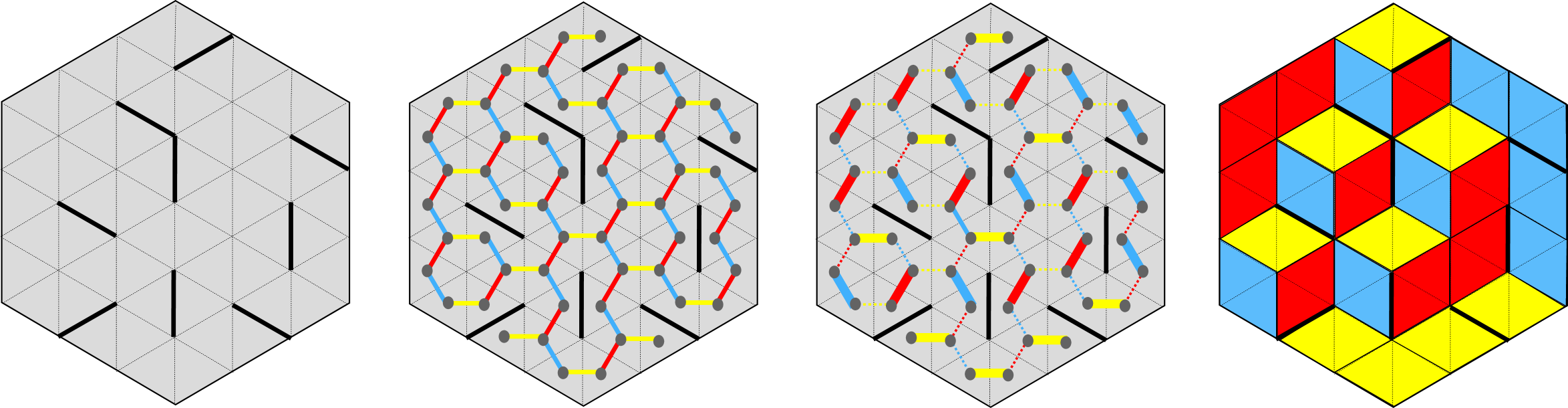}
	\end{center}
	\caption{\label{mat} \textbf{Try to solve a puzzle through a matching computation}. Left, an instance of the Calisson puzzle. In the center, the adjacency graph of the triangles and a perfect matching of the triangles. On the right, the resulting tiling  satisfies the non overlapping condition (i) but violates the saliency condition (ii).  }
\end{figure}

\begin{figure}[ht!]
  \begin{center}		\includegraphics[width=0.85\textwidth]{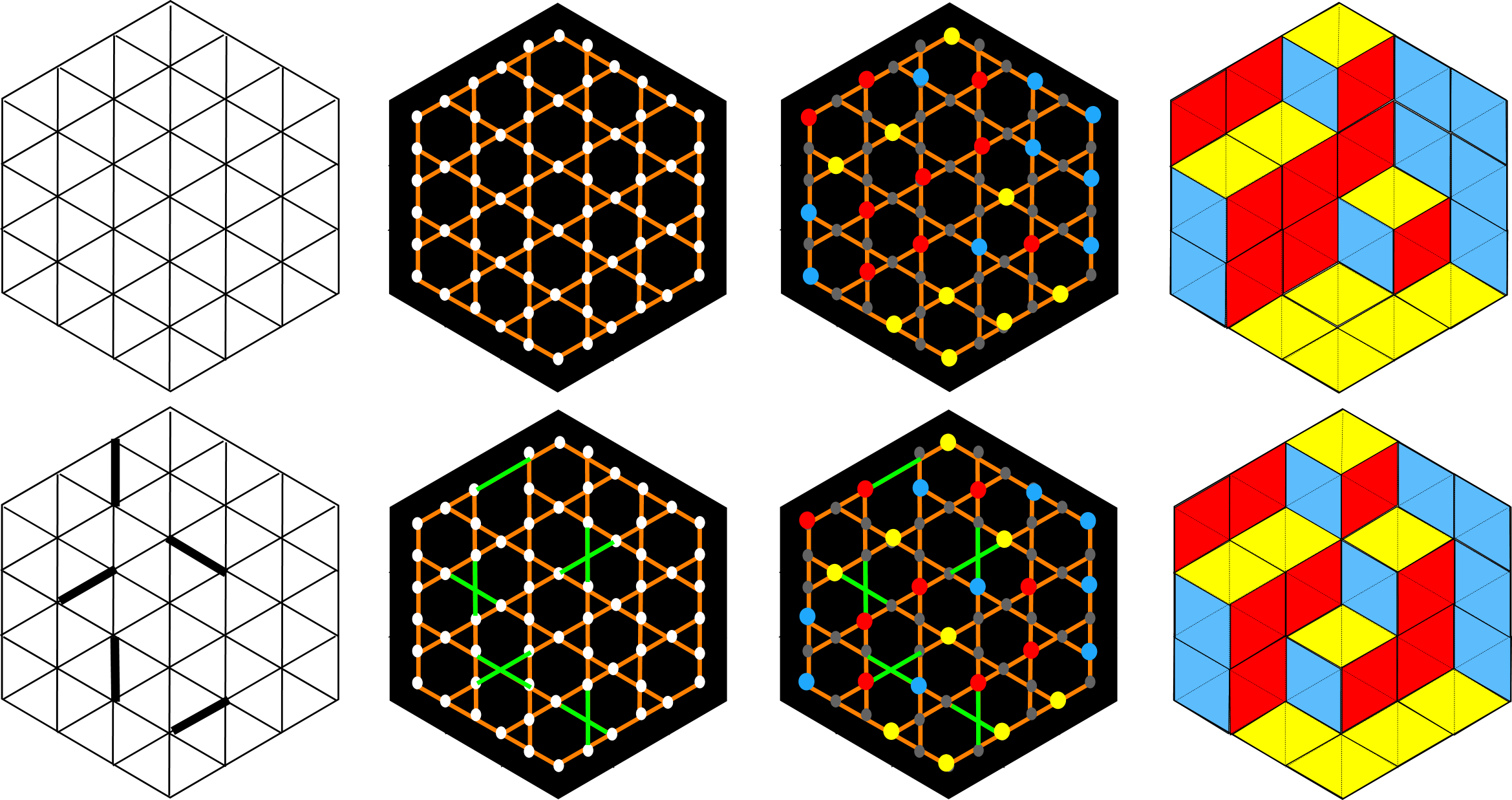}
	\end{center}
	\caption{\label{ind2} \textbf{Reducing Calisson puzzles to maximum independent set.} In the first row, the interior edges of the triangular grid are represented by white nodes. A pair of  edges/nodes is connected iff the pair of edges share a triangle. Then an independent set of white nodes corresponds to a non overlapping set of lozenges. In the second row, we consider an instance of the Calisson puzzle. The puzzle constraints are taken into account by removing the input edges/nodes and reconnecting the orphan nodes (by the green edges). A solution of the Calisson puzzle is given by an independent set of $3n^2$ edge/nodes.  }
\end{figure}

 \textbf{Maximum Independent Set in an almost perfect graph.} The idea is to consider a  dual graph whose nodes are the edges of the triangular grid and where nodes/edges are linked when they belong to a common triangle. In this graph, an independent set of nodes can be seen as a set of lozenges which do not overlap. It reduces the computation of a tiling of a region to the computation of a maximum independent set as done in Fig.~\ref{ind2}. The saliency constraint can be taken into account by replacing the forbidden edges by crossed links (green crosses in the second row of  Fig.~\ref{ind2}). Unfortunatly, this new graph is no more a line perfect graph for which a maximum independent set is known to be computable in polynomial time. 
 
\section{Problem Statement and Results}

We first provide a general formulation of lozenge tiling problems with optional non overlapping and saliency constraints.
We then state the main results of the paper.

\subsection{Problem Statement}

\subparagraph{Regions to which our results apply.}
The triangular grid is denoted $\triangle$. 
We use the letter $R$ to denote a region of the triangular grid.
Formally, $R$ is a simplicial complex whose sets of triangles, edges, and vertices are respectively denoted by $R^2$, $R^1$, and $R^0$.
Motivated by the Calisson puzzle, particular attention is paid to regular hexagonal regions.
The regular hexagon of size $n$ is denoted by $\varhexagon_n$.

The results presented in this paper do not require the region $R$ to be compact: unbounded regions are allowed.
However, we assume that $R$ is simply connected (it is connected and it contains no holes).
In particular, we consider instances where the region $R$ is the entire infinite triangular grid, denoted $R=\triangle$, together with a finite set of input edges.
The boundary of $R$ is denoted by $\partial R$.

Our results can also be extended to regions whose boundary passes several times through the same vertices or edges.
As long as no saliency constraint is imposed on such shared edges, this can be handled by duplicating boundary vertices and edges, as illustrated in Fig.~\ref{duplication}.
For the sake of simplicity, we omit these cases from the formal statements and proofs.

\begin{figure}[ht]
  \begin{center}
    \includegraphics[width=0.65\textwidth]{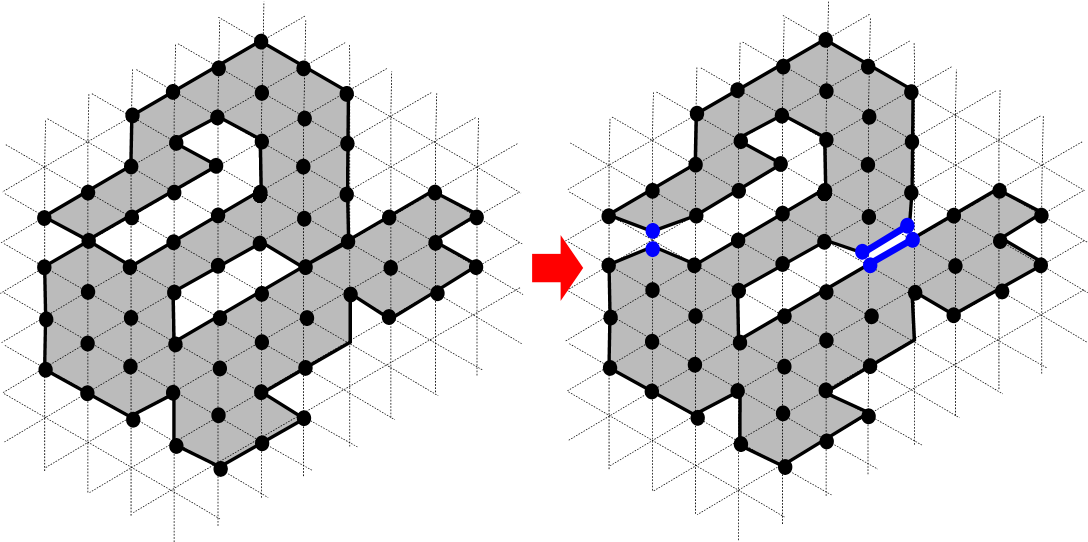}
  \end{center}
  \caption{\label{duplication}
  \textbf{A region $R$} within the scope of our results that requires duplication of boundary vertices and edges to apply our framework.}
\end{figure}

\subparagraph{Edge constraints.}
We consider two types of constraints on edges of the triangular grid.
The first set, denoted $X_1$, consists of edges that must not be overlapped by any lozenge. This is the \emph{non overlapping constraint}.
The second set, denoted $X_2$, consists of edges that must not be overlapped and whose two adjacent lozenges must have distinct orientations. This additional requirement is the \emph{saliency constraint}.

\subparagraph{Generic lozenge tiling problem.}
This leads to the following general tiling problem:

\medskip
\noindent
\texttt{Tiling$(R,X_1,X_2)$}
\vspace{-0.2cm}
\begin{itemize}
    \item \texttt{Input:} A region $R \subseteq \triangle$ and a set $X \subseteq \triangle^1$ of edges of the triangular grid, partitioned into two subsets $X = X_1 \cup X_2$.
    \item \texttt{Output:} A lozenge tiling of the region $R$ such that
    (i) no edge of $X$ is overlapped by a lozenge, and
    (ii) for every edge in $X_2$, the two adjacent lozenges have different orientations.
\end{itemize}

A Calisson puzzle with a set of interior saliency constraints $X_2$ corresponds to the instance
\texttt{Tiling$(\varhexagon_n,\emptyset,X_2)$}.

\subsection{Results}
We present two algorithmic approaches, both rooted in the three-dimensional interpretation of lozenge tilings as stepped surfaces.

\subparagraph{The Advancing Surface Algorithm.}

The advancing surface algorithm is illustrated in Figs.~\ref{start2} and \ref{ads} on an instance of the Calisson puzzle
\texttt{Tiling$(\varhexagon_n,\emptyset,X_2)$}.
It fills an initially empty cube of size $n \times n \times n$ by progressively adding unit cubes,
each time adding as few cubes as possible,
so as to satisfy the non overlapping and saliency constraints
as a simple graph traversal algorithm.
It can be easily extended to bounded, simply connected regions of the triangular grid.
This extension only requires to add a preprocessing step that computes the minimal and maximal tilings
(for instance by running classical Thurston's algorithm \cite{Thurston}).
The advancing surface algorithm then progressively adds unit cubes in the volume between the two stepped surfaces obtained by lifting the two extremal tilings. The algorithm and its analysis are detailed in Subsection \ref{tasa}.

\begin{theorem}\label{wall}
The advancing surface algorithm solves
\texttt{Tiling$(R,X_1,X_2)$}
for a bounded, simply connected region $R$ in running time $O(|\partial R|\cdot |R|)$.
\end{theorem}

In the special case of the Calisson puzzle, the region is the hexagon $\varhexagon_n$. Then the two extremal tilings are straightforward (the projections of the surface of an empty and a full $n\times n \times n$ cube). 

\begin{corollary}\label{maincalissons}
The advancing surface algorithm solves  Calisson puzzle instances\\
\texttt{Tiling$(\varhexagon_n,\emptyset,X_2)$}
with running time $O(n^3)$.
\end{corollary}

The theorem \ref{wall} and Corollary \ref{maincalissons} are proved in Subsection \ref{tasa} after the introduction of a path of consecutive reductions explaining why a so simple algorithm works.

\subparagraph{Reduction of \texttt{Tiling$(R,X_1,X_2)$} to a System of Difference Constraints.}

The main contribution of this paper is the completion of Thurston’s theory of lozenge tilings with a graph-theoretical layer.
This additional structure allows us to translate the tiling problem
\texttt{Tiling$(R,X_1,X_2)$}
into an equivalent system of difference constraints. This system
is induced by a weighted directed graph, denoted
$DC(R,X_1,X_2)$ and defined as follows:

\medskip

\begin{figure}[ht]
  \begin{center}
		\includegraphics[width=\textwidth]{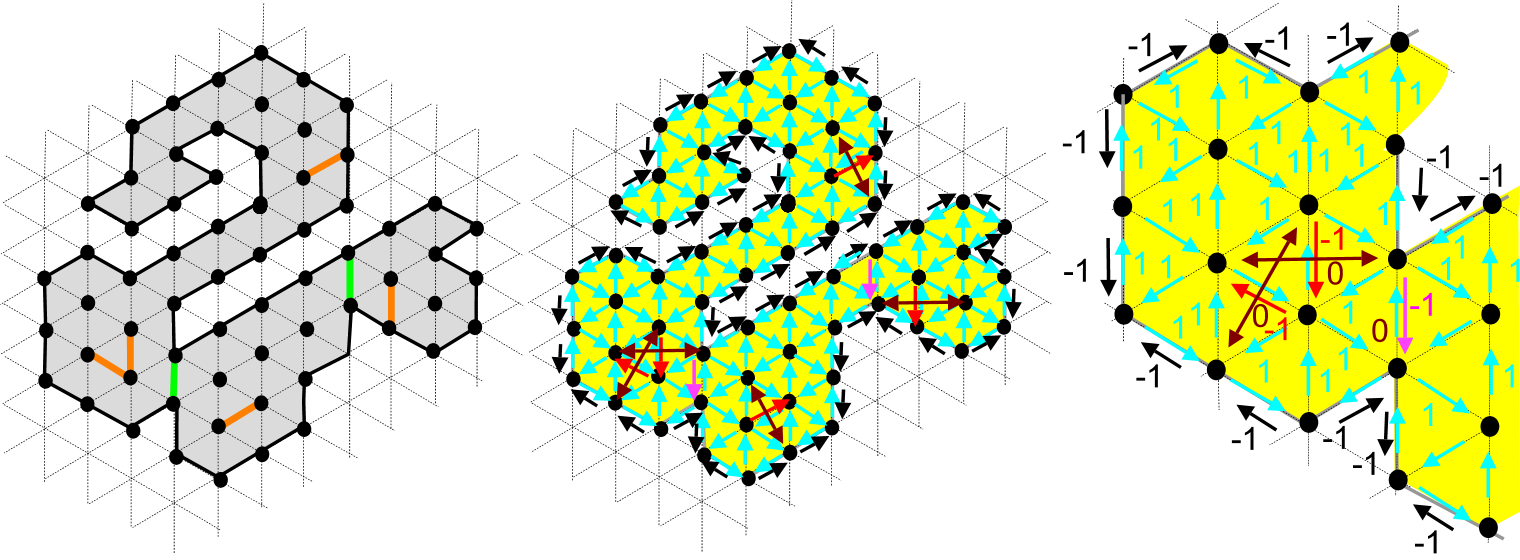}
	\caption{\label{DCgraph} \textbf{Construction of the graph $DC(R,X_1,X_2)$.}
	All the edges of the region in positive direction are weighted by $+1$. 
	Each edges $e$ of $\partial X$ (in black)  or of $X_1$ (in green) generates an edge directed in a negative direction of $e$ weighted by $-1$. Each edge $e_2$ of $X_2$ (drawn in orange) generates an edge directed in a negative direction of $e_1$ weighted by $-1$ and a pair of lateral edges of weights $0$. }
 \end{center}
\end{figure}

\begin{tcolorbox}[colback=red!5!white,
                  colframe=black!75!black,
                  title=Construction of the DC graph
                 ]
\textbf{\texttt{Some notations:} } An oriented edge $u \rightarrow v$ of the triangular grid is said to be
\emph{positive} if it points in the time direction of $12$, $4$, or $8$ o’clock,
and \emph{negative} if it points in the time direction of $2$, $6$, or $10$ o’clock.

We also introduce \emph{lateral edges}, which complement the edges of the triangular grid.
Given an edge $u,v$ of the triangular grid, let $w$ and $w'$ be the two vertices
adjacent to both $u$ and $v$.
The lateral edges associated with $u,v$ are the two oriented edges
$w \rightarrow w'$ and $w' \rightarrow w$.
\tcblower 
\textbf{\texttt{Construction of the weighted directed graph
$DC(R,X_1,X_2)$:\\ } } 
The construction can be followed in Fig.~\ref{DCgraph}.
The vertex set of $DC(R,X_1,X_2)$ is the set $R^0$ of the vertices of the region $R$.
Its directed edges are defined as follows:
\begin{itemize}
    \item Every positively oriented edge contained in $R$ is assigned weight $+1$.
    \item For every edge $e$ in $X = \partial R \cup X_1 \cup X_2$,
    the negatively oriented version of $e$ is assigned weight $-1$.
    These edges encode the non-overlapping constraints.
    \item For every edge $e \in X_2$, the pair of lateral edges associated with $e$
    is assigned weight $0$.
    These edges encode the saliency constraints.
\end{itemize}
\end{tcolorbox}
\medskip

A system of difference constraints induced by a weighted directed graph
consists in assigning a value $h(v)$ to each vertex $v$
such that, for every directed edge $u \rightarrow v$ with weight $w$,
the inequality
$h(v) - h(u) \leq w$
is satisfied (see Fig.~\ref{dc} or reference CLRS  \cite{intro} section 24.4).

\begin{figure}[ht]
  \begin{center}
		\includegraphics[width=\textwidth]{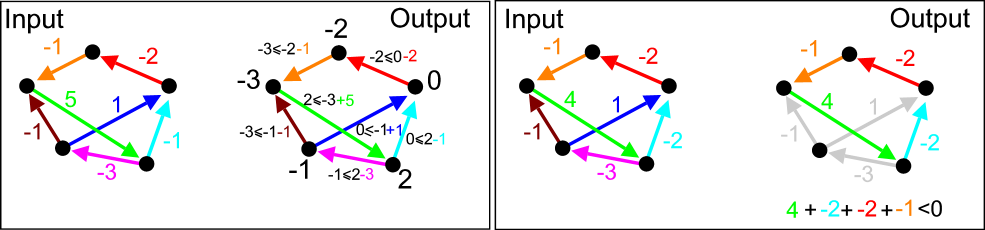}
	\caption{\label{dc} \textbf{Two difference constraints systems.}
	The system induced by the left directed weighted graph admits a solution in black while the right system is not feasible because it contains a cycle with strictly negative total weight.}
 \end{center}
\end{figure}

We claim that the tiling instance
\texttt{Tiling$(R,X_1,X_2)$}
is equivalent to the system of difference constraints induced by
$DC(R,X_1,X_2)$.

\begin{restatable}{theorem}{central}\label{thm:dc}
There is a one-to-one correspondence between tiling which are solutions of
\texttt{Tiling$(R,X_1,X_2)$}
and integer-valued functions
$h : R^0 \rightarrow \mathbb{Z}$
satisfying the system of difference constraints induced by
$DC(R,X_1,X_2)$, up to an additive constant.
\end{restatable}

The tiling is obtained from the integer-valued function $h$ by linking the adjacent vertices of $R$ whose value $f$ differ from $1$ (as illustrated in Figs.~\ref{end} or \ref{start}).

Classically, systems of difference constraints are solved by shortest-path
algorithms such as Bellman--Ford \cite{intro}.
If the graph contains a directed cycle of strictly negative total weight,
then the system has no feasible solution.
Otherwise, shortest-path distances from an arbitrary source yield
a valid solution of the constraints.

This dichotomy is illustrated in Fig.~\ref{start}.

\begin{figure}[ht]
  \begin{center}
		\includegraphics[width=\textwidth]{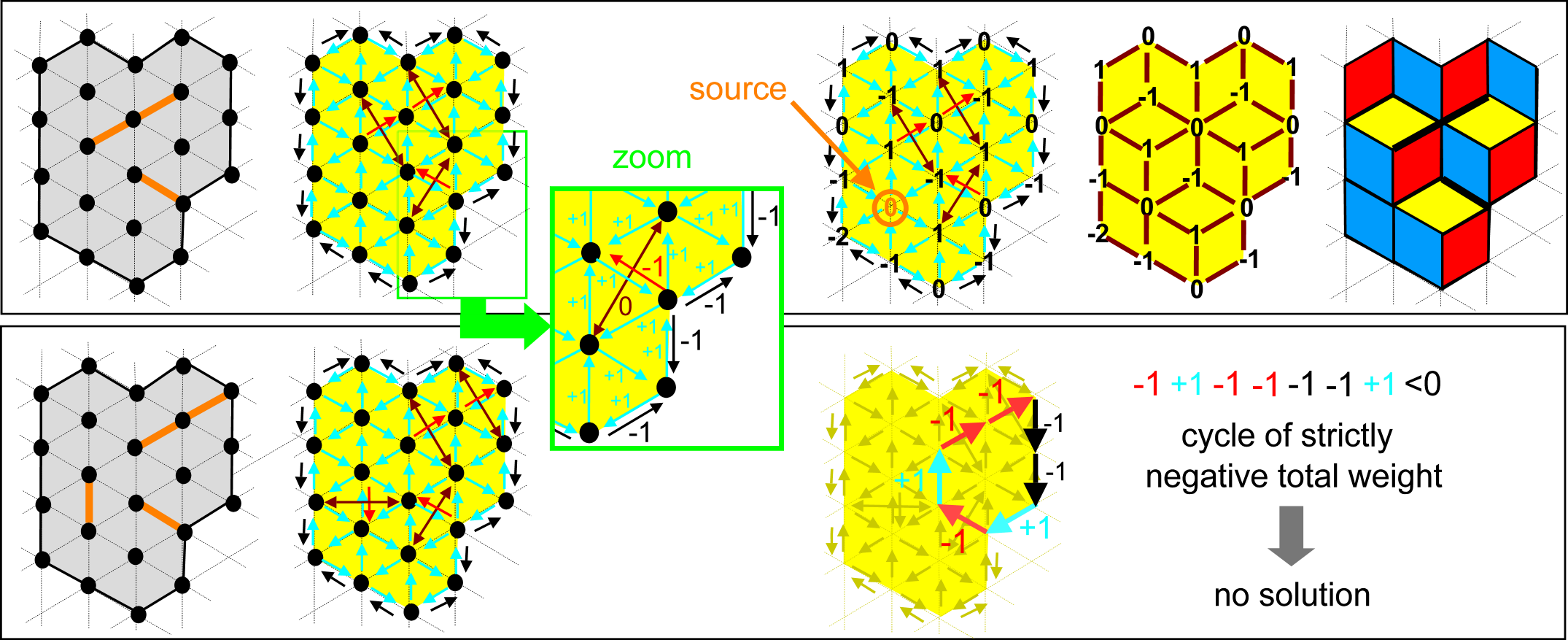}
	\caption{\label{start} \textbf{Tilability characterization.}
	The tiling instance \texttt{Tiling$(R,X_1,X_2)$} is feasible if and only if
	the directed weighted graph $DC(R,X_1,X_2)$ contains no cycle of strictly negative total weight
	(also called an \emph{absorbing cycle}).
	Top: an instance without absorbing cycle. The distances computed by the
	Bellman-Ford algorithm yield a valid tiling.
	Bottom: an instance with a negative cycle. The tiling problem has no solution.}
 \end{center}
\end{figure}

As a direct consequence of Theorem~\ref{thm:dc}, we obtain the following corollary.

\begin{corollary}
For any finite simply connected region $R$,
the tiling problem \texttt{Tiling$(R,X_1,X_2)$}
can be solved in time $O(|R|^2)$
by applying the Bellman--Ford algorithm to the directed graph
$DC(R,X_1,X_2)$.
\end{corollary}

We conclude this section by considering the decision problem of tiling the
entire triangular grid $\triangle$ with a finite set of constraints
$X_1$ and $X_2$.

\begin{corollary}\label{col}
Let $X_1$ and $X_2$ be finite sets of edges.
The feasibility of \texttt{Tiling$(\triangle,X_1,X_2)$}
can be decided by solving a finite system of difference constraints derived
from $DC(\triangle,X_1,X_2)$,
in time $O(|X_1 \cup X_2|^3)$.
\end{corollary}

Due to the absence of boundary conditions and the infinite number of tiles,
this last problem cannot be addressed by classical approaches such as
matching-based methods, Thurston’s original algorithm, SAT encodings,
or maximum independent set formulations.
The corollary \ref{col} illustrates that the difference constraints formulation
induced by $DC(R,X_1,X_2)$ is not a mere technical rewriting of Thurston’s theory,
but rather an essential complement that clarifies the algorithmic theory of
lozenge tilings and enables the solution of tiling problems beyond the reach
of standard methods.

\section{Why it Works}

The goal of this section is to prove Theorem~\ref{thm:dc} by showing that solving the generic tiling problem
\texttt{Tiling$(R,X_1,X_2)$} is equivalent to solving the system of difference constraints induced by the weighted directed graph $DC(R,X_1,X_2)$.

This equivalence is obtained through a sequence of transformations that can be summarized as $
\textit{Tiling} \;\longrightarrow\; \textit{Roof} \;\longrightarrow\; \textit{Directed cut} \;\longrightarrow\; \textit{Height function} \;\longrightarrow\; \textit{Difference constraints}$
Each transformation is the purpose of a dedicated subsection, and an additional subsection is devoted to the treatment of unbreakable edges encoding non-overlapping and saliency constraints. Thurston’s original approach follows the shorter path
$\textit{Tiling} \;\longrightarrow\; \textit{Roof} \;\longrightarrow\; \textit{Height function}
$
which is sufficient for simply connected regions without interior constraints.
Our contribution is to make explicit the intermediate graph-theoretic structure of directed cuts,
which provides a complete and uniform framework before passing to height functions and difference constraints.
This additional layer makes it possible to incorporate interior constraints and to rely on classical
algorithmic tools from Theoretical Computer Science.

\subsection{Notations}

We first introduce a few notations.

\subparagraph{The Grids $\square$, $\triangle$, and the Projection $\varphi$.}
The \emph{primary unit cube} is $[0,1]^3$.
The cubes of the cubic grid are denoted $(x,y,z)+[0,1]^3$, with $(x,y,z) \in \mathbb{Z}^3$ (see Fig.~\ref{grid}). 
The sets of cubes, faces, edges, and vertices of the cubic grid are respectively denoted
$\square^3$, $\square^2$, $\square^1$, and $\square^0$.
Throughout the paper, we follow the convention that a superscript $k \in \{0,1,2,3\}$ indicates the dimension of the objects in the corresponding set.

The infinite triangular grid $\triangle$ is obtained by projecting the edges of $\square^1$ using a map $\varphi$.
More precisely, $\varphi$ denotes the projection of the three-dimensional space $\mathbb{R}^3$ onto a plane $H$
of equation $x+y+z=h$, along the direction $\mathds{1}=(1,1,1)$.
The plane $H$ is naturally decomposed into a simplicial complex
$\triangle = \triangle^2 \cup \triangle^1 \cup \triangle^0$ consisting of triangles, edges, and vertices.

Any vertex of $\triangle^0$ is incident to six edges.
Their directions are given by $\varphi(1,0,0)$, $\varphi(0,1,0)$, $\varphi(0,0,1)$, and their opposites.
Rather than using two-dimensional coordinates in the plane $H$, we use \emph{homogeneous coordinates}:
a point of $H$ is represented as $\varphi(x,y,z)$ with $(x,y,z) \in \mathbb{R}^3$.
Clearly, $\varphi(x,y,z) = \varphi(x+k,y+k,z+k)$ for any $k \in \mathbb{R}$.
Adding such a constant changes the depth in the $\mathds{1}$ direction without affecting the projection.

This notion of depth was introduced by W.~Thurston under the name of \emph{height}.
We adopt this terminology and define the height as $x+y+z$, equivalently as the depth in the $\mathds{1}$ direction.

\begin{figure}[ht]
\begin{center}
    \includegraphics[width=0.50\textwidth]{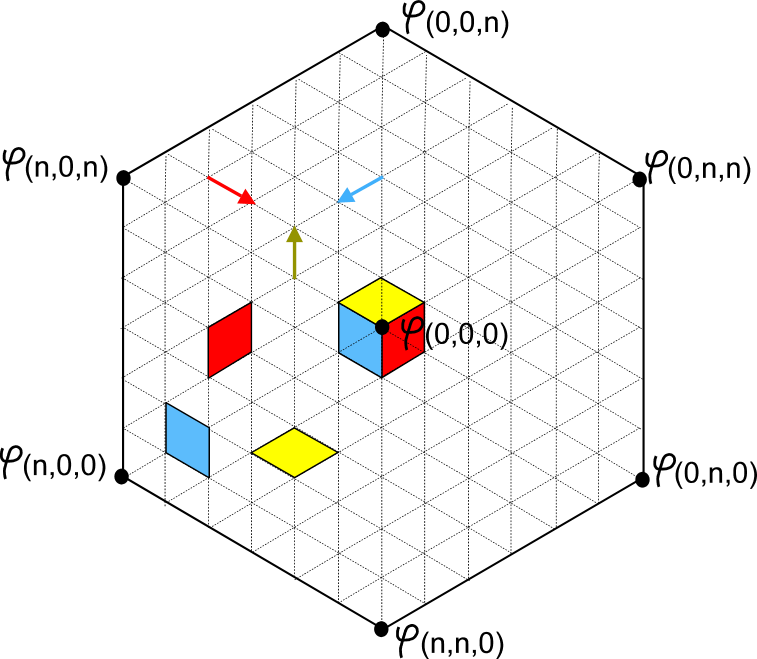}
\end{center}
\caption{\label{grid}
\textbf{The hexagonal region $\varhexagon_n$ of the triangular grid $\triangle$}
used in Calisson puzzles of size $n=6$.
The vertex $\varphi(0,0,0)$ is at the center, surrounded by the projection $\varphi(C)$
of the primary cube $C=[0,1]^3$.
We also show three edges of $\triangle^1$ in the directions
$\varphi(1,0,0)$, $\varphi(0,1,0)$, and $\varphi(0,0,1)$,
as well as a yellow, a red, and a blue lozenge.}
\end{figure}

\subparagraph{The Lozenges.}
Our two-dimensional tiles are lozenges.
A \emph{lozenge} is defined as the projection under $\varphi$ of a two-dimensional face of the cubic grid $\square$.
Since the faces of $\square$ have three possible orientations,
there are exactly three types of lozenges.
Blue, red, and yellow lozenges are respectively the projections of square faces
with normal directions $(1,0,0)$, $(0,1,0)$, and $(0,0,1)$.
\subsection{From Tilings to Roofs (Thurston's Theorem)}

William Thurston showed that any lozenge tiling $T$ of a simply connected region $R$
can be lifted through $\varphi^{-1}$ to a monotone stepped surface of the cubic grid $\square$,
where monotonicity is considered with respect to the direction $\mathds{1}$.
This surface is unique up to translations in the $\mathds{1}$ direction.
By fixing the height of a single point, the lifting $\varphi^{-1}$ is therefore uniquely defined.
In~\cite{Thurston,Thurston2}, this surface $\varphi^{-1}(T)$ is called the \emph{roof} of the tiling $T$.

\begin{theorem}[Thurston]\label{thurstontheorem}
Any lozenge tiling $T$ of a simply connected region $R$ of the triangular grid $\triangle$
can be lifted to a surface of the cubical $2$-complex, denoted $\varphi^{-1}(T)$ and called the roof of $T$,
such that the projections of the square faces of the roof are exactly the lozenges of $T$.
\end{theorem}

Moreover, the height of a vertex $v \in R^0$ is  defined as the value $x+y+z$
of the unique point in $\varphi^{-1}(v)$ that belongs to the roof $\varphi^{-1}(T)$.
This height is defined up to an additive constant,
but height differences are unambiguous and play a central role in what follows.

\begin{figure}[ht]
  \begin{center}		\includegraphics[width=0.25\textwidth]{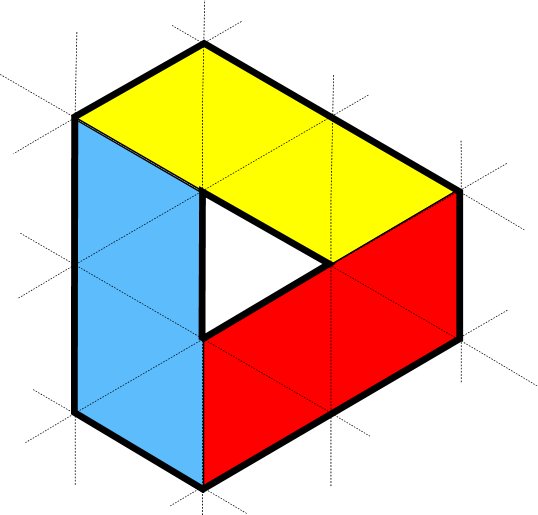}
	\end{center}
	\caption{\label{excluded} \textbf{An excluded region}. The regions with holes are excluded from Thurston theorem because they admit tilings which can not be lifted in monotone surfaces of the cubic complex $\square$.
    }
\end{figure}
\vspace{-0.2cm}

\subsection{From Roofs to Directed Cuts}

At this point, we depart from the standard geometric viewpoint introduced by Thurston to adopt a graph Theoretical Computer Science framework (directed cuts and difference constraints systems) that is best suited to address
the algorithmic questions arising in the computation of lozenge tilings with local constraints.

\subparagraph{The Ascendant Graph $\H_R$ of Cubes and Its Projection.}

With our notation, each cube $(x,y,z)+[0,1]^3$ of $\square^3$ is identified with its base point $(x,y,z)$,
yielding a natural one-to-one correspondence between $\square^3$ and $\mathbb{Z}^3$.
We define $\square_R^3 \subset \square^3$ as the set of cubes whose base points project into the region $R$,
that is, such that $\varphi(x,y,z) \in R$.
Equivalently, $\square_R^3$ consists of infinite stacks of cubes $\varphi ^{-1} (\varphi(x,y,z))$ above the vertices $\varphi(x,y,z) \in R$
in the direction $\mathds{1}$.

We endow $\square_R^3$ with a directed graph structure, called the \emph{ascendant graph} and denoted
$\H_R=(\square_R^3,\wedge_R)$.
There is a directed edge $(x,y,z)+[0,1]^3 \;\rightarrow\; (x',y',z')+[0,1]^3$
in $\wedge_R$ if and only if the following two conditions hold:
\begin{itemize}
    \item[(a)] the cube $(x',y',z')+[0,1]^3$ is one of the three cubes $(x+1,y,z)+[0,1]^3$, $(x,y+1,z)+[0,1]^3$, or $(x,y,z+1)+[0,1]^3$.
    In other words, edges of $\H_R$ go from a cube to one of its three face-adjacent cubes
    of height exactly one unit higher.
    \item[(b)] The projection of this edge by $\varphi$,
    $\varphi(x,y,z)\rightarrow\varphi(x',y',z')$,
    is an edge of the triangular grid belonging to the region,
    that is, $\varphi(e)\in \triangle_R^1$.
\end{itemize}

When the boundary of the region folds back onto itself,
Condition~(b) prevents the introduction of spurious edges.

By construction, the edges of the ascendant graph $\H_R$ correspond to square faces of the cubic grid
whose projections by $\varphi$ are either entirely contained in $R$
or overlap boundary edges of $R$.
This observation will play a central role in the sequel.

\begin{remark}\label{remark}
The projection of the adjacent square face/edge of $\H_R$ from
$(x,y,z)+[0,1]^3$ to $(x+1,y,z)+[0,1]^3$
is precisely the lozenge overlapping the edge
$\varphi(x,y,z),\varphi(x+1,y,z)$ of the triangular grid.
The same remark holds for the two other coordinate directions.
\end{remark}

The ascendant graph $\H_R$ is acyclic, hence a directed acyclic graph (DAG).
Moreover, it is invariant under translations by vectors of the form $(k,k,k)$ with $k\in\mathbb{Z}$,
and therefore admits a $\mathbb{Z}$-periodic structure.
This periodicity is essential and will be exploited later.

\subparagraph{Roofs as Directed Cuts of the Ascendant Graph.}
We now return to the roof associated with a lozenge tiling $T$ of a simply connected region $R$
and interpret it as a directed cut of the ascendant graph $\H_R$.

Each square face of the roof is adjacent to a pair of cubes
$(x,y,z)+[0,1]^3$ and $(x',y',z')+[0,1]^3$,
where $(x',y',z')$ equals one of $(x+1,y,z)$, $(x,y+1,z)$, or $(x,y,z+1)$.
Since the projection of this square face by $\varphi$ lies in the region $R$,
both cubes belong to $\square_R^3$ and therefore define an edge of $\H_R$.
It follows that the square faces of the roof are a subset of edges of the ascendant graph. This subset has an important property.

We now introduce the notion of a \emph{directed cut} (or \emph{dicut}) of a directed graph.
A dicut is a partition of the vertex set into two non empty subsets that we denote $\bottom$ and $\top$
such that all edges crossing the cut are directed from $\bottom$ to $\top$. The difference between a simple cut and a dicut is illustrated in Fig.~\ref{dicut}.

\begin{figure}[ht]
  \begin{center}		\includegraphics[width=0.80\textwidth]{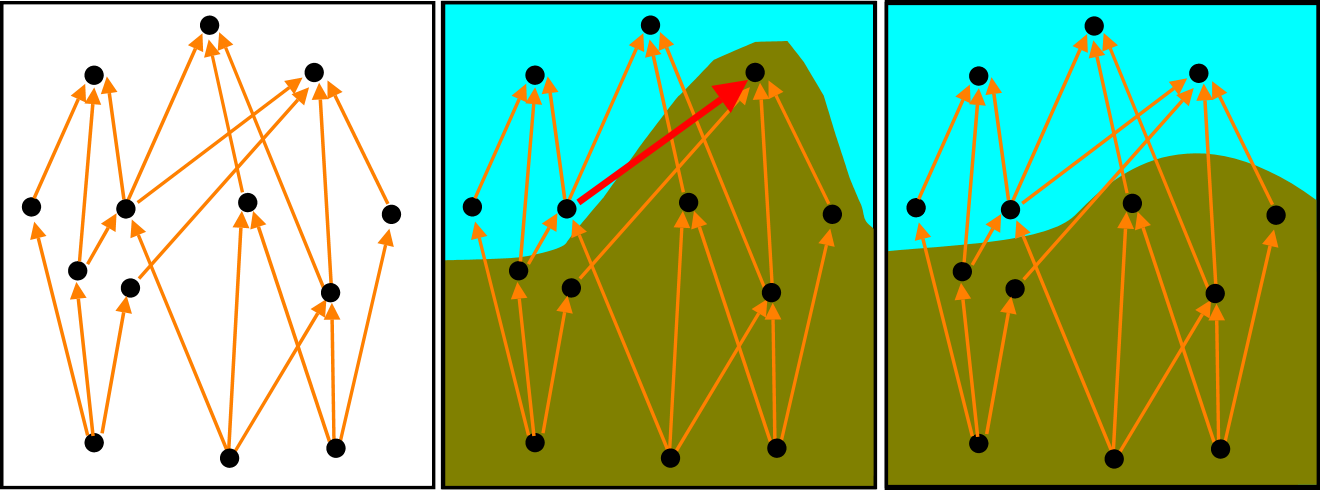}
	\end{center}
	\caption{\label{dicut} \textbf{Cuts and dicuts.} On the left, a directed graph. In the middle, a cut is drawn. It is not a dicut because the edges from one part to the other are not uniformly directed due to the red directed edge. On the right, a dicut with $\bottom$ in green and $\top$ in blue.}
\end{figure}
One of the properties that we use in the following is that a directed path in a directed graph cannot be cut twice by a dicut.

The roof of the tiling $T$ being a monotone surface, it provides a partition of the cubes of $\square _R$ which are above and below it. This partition is a dicut of $\H_R$. 

\begin{claim}
The roof of a tiling $T$ of a simply connected region $R$ defines a dicut of the ascendant graph $\H_R$.
\end{claim}

At this stage, the correspondence is only one-way.
While the roof of a tiling of $R$ always induces a directed cut of the ascendant graph $\H_R$,
the converse is not true in general:
there exist dicuts of $\H_R$ whose projections
(the projections of the square faces separating the two sides of the cut)
do not form a valid tiling of the region $R$,
because some projected lozenges may overlap boundary edges of $R$
or violate interior constraints.

From the algorithmic perspective of solving an instance of
\texttt{Tiling$(R,X_1,X_2)$},
the central question therefore becomes the following:
\emph{under which conditions does the projection of a dicut of the ascendant graph $\H_R$
yield a valid tiling of the region $R$ satisfying the non overlapping and saliency constraints
encoded by $X_1$ and $X_2$?}

\subsection{Encoding the Local Constraints by Unbreakable Edges}

We arrive at the step where the dicut perspective is essential.
We express  the constraints that a dicut must satisfy for providing a valid solution of \texttt{Tiling$(R,X_1,X_2)$} by introducing  the notion of \emph{unbreakable edges}.
Certain edges of the ascendant graph must not be crossed by a dicut,
because cutting them would produce projected lozenges
that violate the tiling constraints.
We now explain how the non overlapping and saliency constraints
translate into unbreakable edges in $\H_R$.

\subparagraph{Non Overlapping Constraints.}

The first type of constraint enforces that no lozenge overlaps
either a boundary edge of $\partial R$
or an interior edge belonging to $X=X_1\cup X_2$.
Consider, for instance, an edge of the triangular grid
between the vertices $\varphi(x,y,z)$ and $\varphi(x,y,z+1)$
that must not be overlapped by a lozenge.
By Remark~\ref{remark}, the lozenge overlapping this edge
is exactly the projection of the square face separating the cubes
$(x,y,z)+[0,1]^3$ and $(x,y,z+1)+[0,1]^3$.

Therefore, forbidding this overlap is equivalent to requiring
that the directed edge of $\H_R$ from
$(x,y,z)+[0,1]^3$ to $(x,y,z+1)+[0,1]^3$
does not belong to the dicut.
The same restriction applies to all translated pairs
$
(x+k,y+k,z+k)+[0,1]^3$ and $(x+k,y+k,z+k+1)+[0,1]^3$ for any $k\in\mathbb{Z}
$.
In other words, all directed edges between these pairs of cubes
are declared \emph{unbreakable}.
Conversely, if a directed cut does not cut any of these unbreakable edges,
then the corresponding lozenge overlapping the edge
$\varphi(x,y,z),\varphi(x,y,z+1)$
cannot appear in the projection. Equivalent claims hold for the two other directions. It shows that by declaring unbreakable some pairs of cubes $
(x+k,y+k,z+k)+[0,1]^3$ and $(x+k,y+k,z+k+1)+[0,1]^3$ for any $k\in\mathbb{Z}
$, we exactly encode the non overlapping constraints.

\subparagraph{Saliency Constraints.}

We now consider the saliency constraints imposed by the edges of $X_2$.
Assume that the edge $e$ between the vertices
$\varphi(x,y,z)$ and $\varphi(x,y,z+1)$ belongs to $X_2$.
The saliency constraint requires that the two lozenges adjacent to $e$
have distinct orientations.

To encode this condition in the ascendant graph,
we introduce four families of cubes for each integer $k$:
\[
\begin{aligned}
L_k &= (x+k,\,y+k-1,\,z+k)+[0,1]^3,\\
R_k &= (x+k-1,\,y+k,\,z+k)+[0,1]^3,\\
F_k &= (x+k,\,y+k,\,z+k)+[0,1]^3,\\
B_k &= (x+k-1,\,y+k-1,\,z+k)+[0,1]^3,
\end{aligned}
\]
as illustrated in Fig.~\ref{LRBF}.
Each of these cubes has a face whose projection is a lozenge
adjacent to or overlapping the edge $e$.

\begin{figure}[tb]
  \begin{center}
		\includegraphics[width=0.55\textwidth]{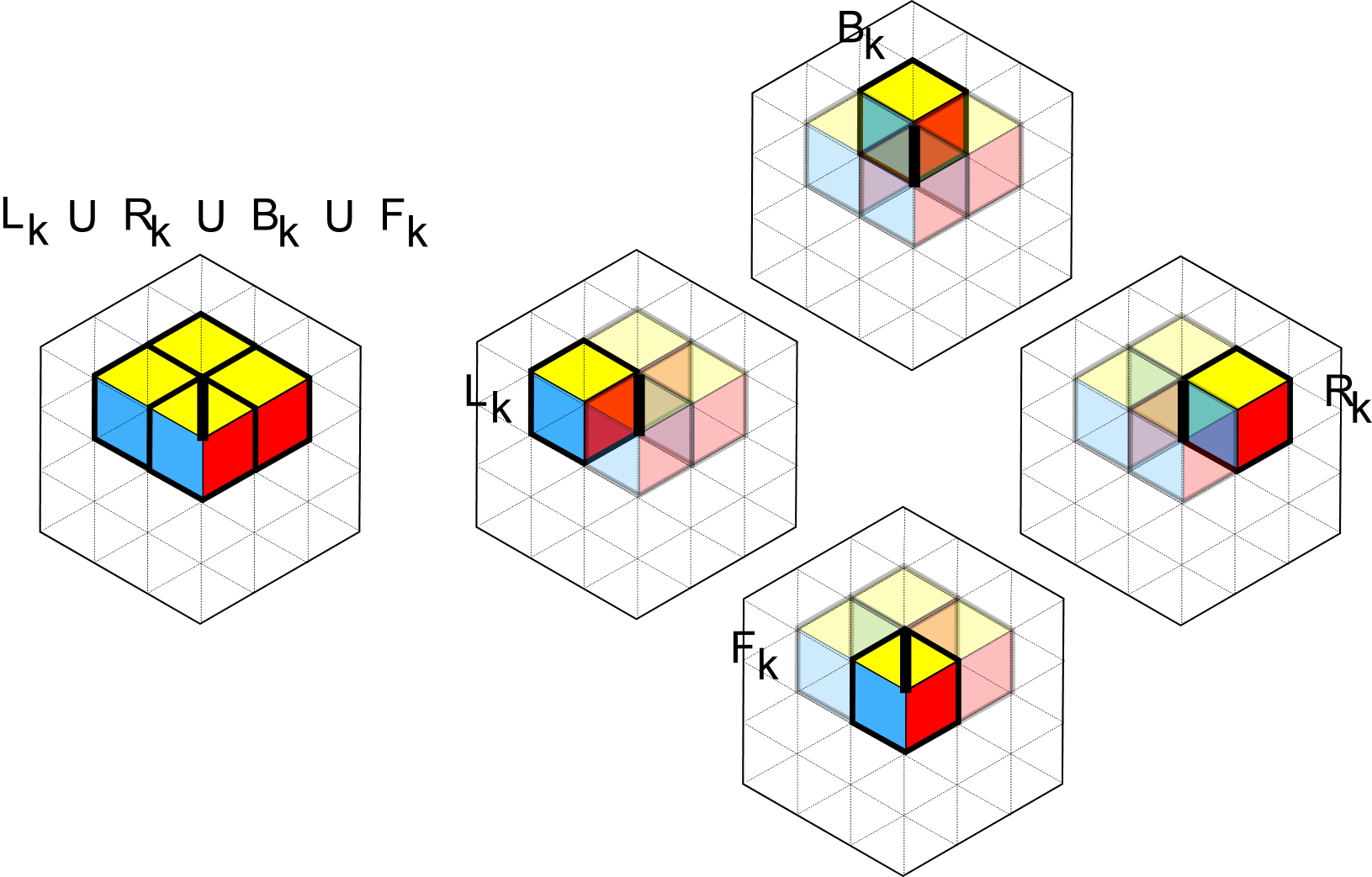}
	\end{center}
	\caption{\label{LRBF} \textbf{The $B_k$, $L_k$, $R_k$ and $F_k$ cubes} for a given integer $k$ around an edge $e$ of the triangular grid.}
\end{figure}

With this notation, the non overlapping constraint on $e$
is expressed by declaring the edge from $F_k$ to $B_{k+1}$ unbreakable. 
We now show how the saliency constraint itself can be enforced.
The ascendant graph $\H_R$ contains the following chain of cubes:
$\cdots \;\rightarrow\; B_k \;\rightarrow\; L_k \;\text{ and }\; R_k
\;\rightarrow\; F_k \;\rightarrow\; B_{k+1} \;\rightarrow\; \cdots$
Any directed cut must intersect this infinite periodic chain at least once,
otherwise one side of the cut would be empty.
Several cutting patterns are possible, as shown in Fig.~\ref{fourcases}.

Cutting the two edges $B_k\rightarrow L_k$ and $B_k\rightarrow R_k$
satisfies the saliency constraint
(Case~1 in Fig.~\ref{fourcases}).
In contrast, cutting $B_k\rightarrow R_k$ together with $L_k\rightarrow F_k$
(Case~2), or symmetrically cutting $B_k\rightarrow L_k$ together with
$R_k\rightarrow F_k$ (Case~$2'$),
produces two adjacent lozenges of the same orientation
and therefore violates the saliency constraint.
Cutting both edges $L_k\rightarrow F_k$ and $R_k\rightarrow F_k$
is valid (Case~3).
Finally, cutting $F_k\rightarrow B_{k+1}$ violates
the non overlapping constraint on $e$.

\begin{figure}[tb]
  \begin{center}
		\includegraphics[width=\textwidth]{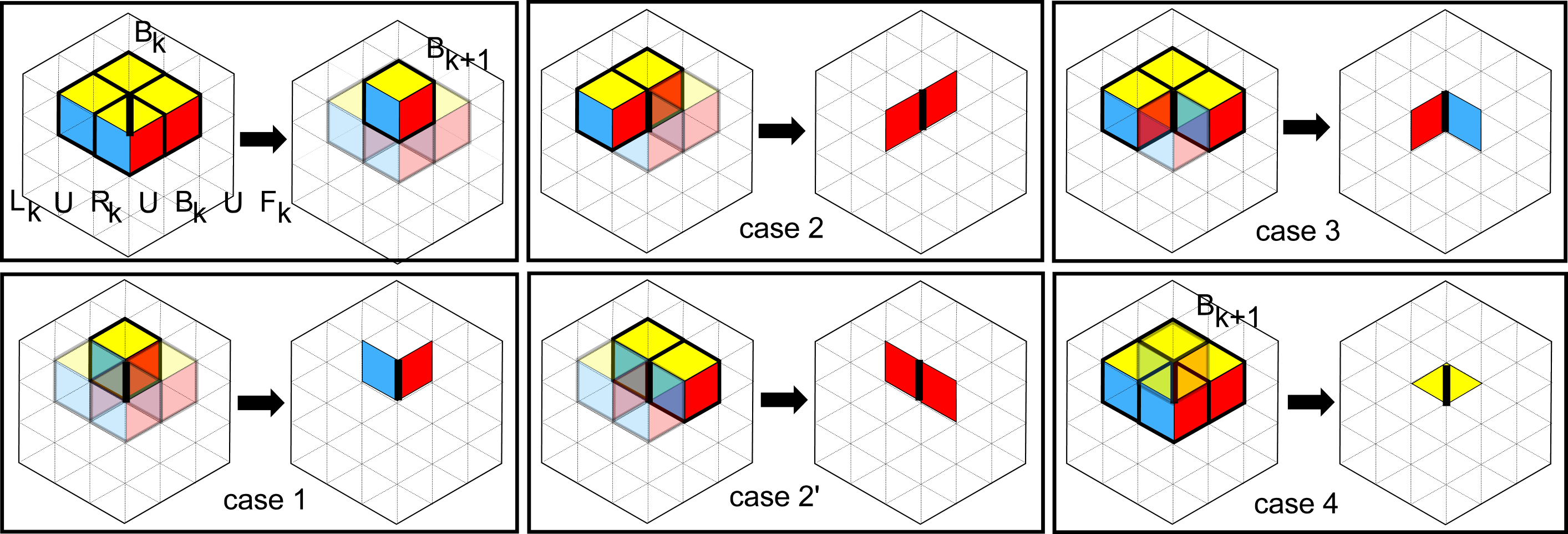}
	\end{center}
	\caption{\label{fourcases} \textbf{The different potential cuts of the ascendant chain $ B_k \rightarrow L_k / R_k
\rightarrow F_k \rightarrow B_{k+1} $ and the corresponding tiling configurations} }
\end{figure}

The conclusion is that the saliency constraint on the edge $e$
is satisfied if and only if the two cubes $L_k$ and $R_k$
lie on the same side of the dicut.
Equivalently, for every $k\in\mathbb{Z}$,
the pair $L_k,R_k$ must not be separated by the dicut.
Thus, the saliency constraint is enforced by declaring
an unbreakable edge between $L_k$ and $R_k$ for all $k$.

\subparagraph{The Enhanced Ascendant Graph $\H_R^{X_1,X_2}$.}

A simple way to enforce that two vertices $u$ and $v$ of a directed graph
belong to the same side of any dicut
is to add a pair of opposite directed edges:
one from $u$ to $v$ and one from $v$ to $u$.
Indeed, no dicut can separate $u$ and $v$ without cutting at least one of these edges in a wrong direction.
We refer to such pairs as \emph{unbreakable edges}.

This observation provides a simple mechanism for incorporating
non overlapping and saliency constraints into the ascendant graph.
We denote by $\H_R^{X_1,X_2}$ the \emph{enhanced ascendant graph},
obtained from $\H_R$ by adding, for every unbreakable constraint induced by
$\partial R$, $X_1$, and $X_2$, a pair of edges in opposite directions
between the corresponding cubes.

We can now state the central correspondence between tilings and dicuts.

\begin{proposition}\label{propi}
A tiling $T$ is a solution of the generic tiling problem
\texttt{Tiling$(R,X_1,X_2)$}
if and only if $T$ is the projection by $\varphi$
of a dicut of the enhanced ascendant graph $\H_R^{X_1,X_2}$.
\end{proposition}

Moreover, the corresponding dicut is unique up to a global translation
along the vector $\mathds{1}$, reflecting the classical height-shift ambiguity.

\begin{proof}
Most of the work has already been done.
By Thurston's theorem, any tiling solution $T$ of the region $R$
can be lifted to a roof, which is a dicut of the ascendant graph $\H_R$.
Since the projection of $T$ satisfies the non overlapping and saliency constraints,
this dicut does not cross any unbreakable edge,
and therefore is also a dicut of the enhanced graph $\H_R^{X_1,X_2}$.

Conversely, let us consider a dicut of $\H_R^{X_1,X_2}$.
We have shown that the unbreakable edges enforce the non overlapping
and saliency constraints.
It remains to prove that the projection of the cut defines a valid tiling of $R$,
that is, that every triangle of the region is covered by exactly one lozenge.

First, assume that a triangle of $R$ is not covered.
For instance, consider the triangle with vertices
$\varphi(x-1,y,z)$, $\varphi(x,y,z)$, and $\varphi(x,y,z+1)$.
In this case, none of the edges of the infinite ascendant chain
\[
\cdots \rightarrow (x-1,y,z)+[0,1]^3
\rightarrow (x,y,z)+[0,1]^3
\rightarrow (x,y,z+1)+[0,1]^3
\rightarrow (x,y+1,z+1)+[0,1]^3
\rightarrow \cdots
\]
is cut.
This implies that one side of the dicut is empty,
which contradicts the definition of a dicut.
Hence, every triangle of $R$ is covered by at least one lozenge.

It remains to show that no triangle is covered by more than one lozenge.
Assume that the same triangle
$\varphi(x-1,y,z)$, $\varphi(x,y,z)$, $\varphi(x,y,z+1)$
is covered twice.
Then the corresponding ascendant chain above
is cut at least twice.
This contradicts the property any directed path can be crossed at most once by a dicut.
Therefore, each triangle of $R$ is covered by exactly one lozenge,
and the projection of the dicut defines a valid tiling.
\end{proof}

The goal of the next steps is to transform the problem of computing a dicut in the periodic graph $\H_R^{X_1,X_2}$ into a system of difference constraints.
This type of reduction is not specific to the graph $\H_R^{X_1,X_2}$: it applies more generally to $\mathbb{Z}$-periodic directed graphs.
Since a full general treatment is beyond the scope of this paper, we only illustrate the mechanism in Fig.~\ref{cutcut}, using $\mathbb{Z}$-periodic graphs different from $\H_R^{X_1,X_2}$.
The purpose of this illustration is to provide intuition for why dicuts in $\mathbb{Z}$-periodic graphs can be characterized and computed via systems of difference constraints.

\begin{figure}[ht]
  \begin{center}		\includegraphics[width=\textwidth]{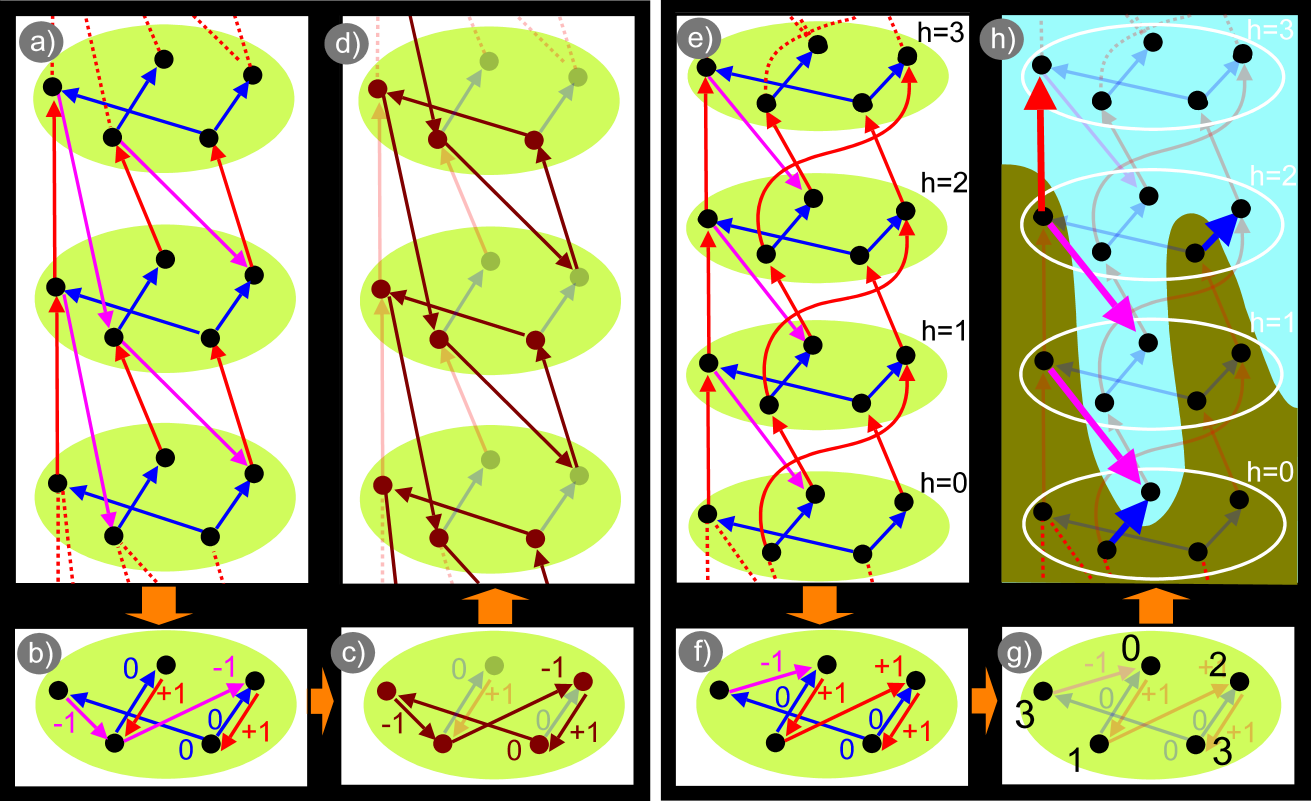}
	\end{center}
	\caption{\label{cutcut} \textbf{Computing a dicut in a $\Z$-periodic graph can be reduced to difference constraints systems.} In a) and e), an ascendant $\Z$-periodic graph $G$ decomposed in layers of different heights. In b) and f), the graph is projected into one layer and each edge is  weighted by the height difference between its target and source. In c) the difference constraints system has a cycle of strictly negative total weight. Then in this case, the graph $G$ has no dicut, as shown in d), due to the infinite descending path. In g), there is no absorbing cycle. The difference system admits solutions which can be lifted in a dicut.}
\end{figure}

\subsection{From Dicuts to Height Functions}

Any dicut of the enhanced ascendant graph
$\H_R^{X_1,X_2}$ can be characterized by the heights of the highest cubes
belonging to the $\bottom$ part of the cut.
Equivalently, a dicut induces an integer-valued function
$
h : R^0 \rightarrow \mathbb{Z}
$,
defined on the vertices $v$ of the triangular grid, where $h(v)$ is the maximum
height of a cube of $\varphi^{-1}(v)$ and that lies below the cut (with the association cube $(x,y,z)+[0,1]^3$ /  point $(x,y,z)$ that we use, the height of the cube $(x,y,z)+[0,1]^3$  is the height of the point $(x,y,z)$).

The function $h$ is defined up to an additive constant, reflecting the invariance
of $\H_R^{X_1,X_2}$ under translation by the vector $(1,1,1)$.
Up to this global shift, $h$ coincides with Thurston’s classical height function
associated with lozenge tilings.

\subsection{From Height Functions to Difference Constraints}

Each directed edge $c_u \rightarrow c_v$ of the enhanced ascendant graph
$\H_R^{X_1,X_2}$ induces a local constraint on the height function $h$.
Here $c_u$ (resp.\ $c_v$) is a cube projecting onto a vertex $u$ (resp.\ $v$)
of the region $R$.
We denote by $w(u,v)$ the fixed height increment associated with this edge,
namely
$$
w(u,v) = height(c_v) - height(c_u).
$$

Let $c'_v$ be the highest cube based in $\varphi^{-1}(v)$ that lies below the cut.
By definition, its height is $h(v)$.
By the $\mathbb{Z}$-periodicity of the graph $\H_R^{X_1,X_2}$,
the directed edge $c_u \rightarrow c_v$ appears at all vertical levels.
In particular, there exists a cube $c'_u$ based in a point of  $\varphi^{-1}(u)$ such that
$c'_u \rightarrow c'_v$ is an edge of $\H_R^{X_1,X_2}$.
Since this edge has height increment $w(u,v)$, the cube $c'_u$ has height
$h(v) - w(u,v)$.

As $c'_u$ lies below the cut, we must have
$
h(u) \geq h(v) - w(u,v)
$, or equivalently,
$$
h(v) - h(u) \le w(u,v).
$$
Thus, each directed edge of $\H_R^{X_1,X_2}$ yields a difference constraint
on the height function $h$.

In the enhanced ascendant graph $\H_R^{X_1,X_2}$, the value of $w(u,v)$ depends
only on the type of the edge.
There are three types of edges:
\begin{enumerate}
    \item The original ascendant edges of $\H_R$, which go from a cube to one of
    its three upper adjacent cubes. These edges increase the height by $1$ and
    therefore have weight $+1$.
    
    \item The descending edges added to encode non overlapping constraints through unbreakable edges
    (arising from $\partial R$ and from $X = X_1 \cup X_2$). These edges go from
    $(x,y,z)+[0,1]^3$ to $(x-1,y,z)+[0,1]^3$, $(x,y-1,z)+[0,1]^3$, or
    $(x,y,z-1)+[0,1]^3$, and decrease the height by $1$. Their weight is $-1$.
    
    \item The unbreakable edges encoding saliency constraints, which connect
    pairs of cubes at the same height (such as the cubes $L_k$ and $R_k$).
    These edges have weight $0$.
\end{enumerate}

Collecting all such inequalities yields a system of difference constraints whose
variables are the values $h(v)$ for $v \in R^0$.
In Section 2, we have denoted by $DC(R,X_1,X_2)$ the weighted directed graph encoding these constraints.
This last step 
ends the sequence of transformations of the  generic tiling problem \texttt{Tiling$(R,X_1,X_2)$} and  proves Theorem \ref{thm:dc}

\central*

\subparagraph{Absorbing Cycles and Feasibility.}

A classical result on systems of difference constraints
states that such a system is feasible if and only if
the associated weighted directed graph contains no cycle
of strictly negative total weight. These two cases are illustrated in Fig.~\ref{start}.
It shows the next corollary of Theorem \ref{thm:dc}.

\begin{corollary}
A tiling instance
\texttt{Tiling$(R,X_1,X_2)$}
admits a solution if and only if the graph $DC(R,X_1,X_2)$
has no cycle of strictly negative weight.
\end{corollary}

When there is no cycle of strictly negative weight, a feasible height function
can be computed using classical shortest-path algorithms
for graphs with possibly negative weights, such as the Bellman-Ford algorithm \cite{Bellman}.
The resulting height function directly defines a dicut of
$\H_R^{X_1,X_2}$ and therefore a valid tiling of the region $R$. It provides the first algorithm that we now discuss with its variants. 

\section{Algorithms}

We now discuss the algorithms that can be used to solve the system of difference constraints induced by
$DC(R,X_1,X_2)$, which encodes the generic tiling problem
\texttt{Tiling$(R,X_1,X_2)$}.
We start with the Bellman–Ford algorithm, then revisit Thurston’s classical algorithm and present the advancing surface algorithm.
We conclude the section with the special case where the underlying region is the infinite triangular grid subject to finitely many local constraints.

\subsection{Bellman-Ford}

The Bellman–Ford algorithm is the standard method for solving systems of difference constraints \cite{Bellman}.
Given a weighted directed graph, it detects the presence of a cycle of strictly negative total weight, and otherwise computes a feasible assignment of distances satisfying all constraints.
Its time complexity is $O(|V||E|)$ for a graph with vertex set $V$ and edge set $E$.

A direct strategy for solving \texttt{Tiling$(R,X_1,X_2)$} consists of first constructing the weighted directed graph $DC(R,X_1,X_2)$, which can be done in linear time with respect to the size of the region, and then applying the Bellman–Ford algorithm to the resulting system of difference constraints.
This approach is illustrated in Figs.~\ref{start2} and~\ref{end}.

In the graph $DC(R,X_1,X_2)$, the vertex set is $R^0$, and the number of edges is linear in $|R^0|$.
Therefore, running Bellman–Ford on $DC(R,X_1,X_2)$ requires $O(|R^0|^2)$ time or equivalently $O(|R|^2)$ time.
In particular, for the Calisson puzzle
\texttt{Tiling$(\varhexagon_n,\emptyset,X_2)$},
where $|R| = \Theta(n^2)$,
the Bellman–Ford approach runs in time $O(n^4)$.

\subsection{Thurston's Algorithm Revisited}

Thurston’s classical algorithm applies to simply connected, bounded regions without interior constraints, that is, to instances of the form
\texttt{Tiling$(R,\emptyset,\emptyset)$}.
Its efficiency relies on the special structure of the graph $DC(R,\emptyset,\emptyset)$ when $R$ is bounded.

By construction, every boundary edge of $R$ appears in $DC(R,\emptyset,\emptyset)$ with weight $+1$ in the positive orientation and $-1$ in the opposite orientation.
Consequently, if the total weight of the boundary cycle is $x$ in one direction of traversal, it is $-x$ in the other direction. It follows that if the total weight is nonzero in either direction, then the instance is not tilable.

We now assume that the boundary cycle has total weight zero.
In the setting of \texttt{Tiling$(R,\emptyset,\emptyset)$}, the graph $DC(R,\emptyset,\emptyset)$ contains no negatively weighted edges in the interior of the region.

We recall a classical observation underlying Thurston’s algorithm:
if a region $R$ satisfying the previous boundary condition is nevertheless not tilable, then there exists a shortcut between two boundary vertices, meaning a path in the interior whose total weight is strictly smaller than that of the corresponding boundary path.

Indeed, if $R$ is not tilable, Theorem~\ref{thm:dc} implies that $DC(R,\emptyset,\emptyset)$ contains a cycle of strictly negative total weight.
Since all negatively weighted edges lie on the boundary, such a cycle must involve at least two boundary vertices.
Let $(v_i)$ denote the sequence of boundary vertices encountered along this cycle.
For at least one pair $(v_i, v_{i+1})$, the interior path connecting them has strictly smaller weight than the boundary path between the same vertices, yielding a shortcut.

We now come back to Thurston’s algorithm and explain how it works.
First, it computes the heights of the boundary vertices by traversing $\partial R$.
If the total weight of the boundary cycle is nonzero, the algorithm immediately concludes that $R$ is not tilable.
Otherwise, since all interior edges have nonnegative weight, the algorithm computes the remaining distances in $DC(R,\emptyset,\emptyset)$ using a Dijkstra-like shortest-path procedure.
If, during this process, the distance of a boundary vertex is decreased, a shortcut has been detected and the region is not tilable.
If no such decrease occurs, the computed distances/heights define a valid tiling of $R$.

This algorithm runs in time $O(|R|\log |R|)$.

\subsection{The Advancing Surface Algorithm}\label{tasa}

We now present the \emph{advancing surface algorithm} for solving tiling instances
\texttt{Tiling$(R,X_1,X_2)$} with interior non-overlapping and saliency constraints.
Unlike the previous approaches, this algorithm does not compute a height function on the vertices of the region.
Instead, it works directly in the graph of cubes by constructing a dicut of the enhanced ascendant graph
$\H_R^{X_1,X_2}$.

The key idea is that it is not necessary to consider the full infinite graph $\H_R^{X_1,X_2}$.
One can restrict attention to a finite domain delimited by two extremal dicuts and then
compute a solution by a graph traversal procedure.

\subparagraph{First step: Lower and Upper Dicuts.}
As in Thurston’s algorithm, we first examine the boundary of the region $R$.
If the directed boundary cycles has nonzero total weight, then $R$ is not tilable and we stop.
We therefore assume that the boundary cycle has total weight zero.

We temporarily ignore the interior constraints and work in the cube graph
$\H_R^{\emptyset,\emptyset}$.
The boundary cycle of $R$ can be lifted to a directed cycle $c$ in this graph.
By periodicity, this yields an infinite stack of directed cycles, each translated from the previous one
by the vector $\mathds{1}=(1,1,1)$.
If the region $R$ is tilable, the graph $\H_R^{\emptyset,\emptyset}$ admits a dicut, and each cycle of this stack
lies entirely either in the bottom or in the top part of the cut.
Consequently, there exists a pair of consecutive cycles $c_0$ and $c_1$ such that
$c_0$ lies in the bottom part and $c_1$ in the top part of the dicut.
Then 
 all cubes that can reach a cube of $c_0$ by a directed path
belong to the lower part of a new dicut, which we call the \emph{lower cut}.
Dually, all cubes reachable from a cube of $c_1$ belong to the upper part of a second new dicut,
called the \emph{upper cut}.
Any dicut of $\H_R^{X_1,X_2}$ separating $c_0$ and $c_1$ must lie between these two extremal cuts.

The first step of the algorithm therefore consists in computing the lower and upper dicuts
of $\H_R^{\emptyset,\emptyset}$ using Thurston’s algorithm.
Either this step determines that $R$ is not tilable, or it provides the two extremal dicuts.

\subparagraph{Second step: Incorporating Interior Constraints.}
Assuming the first step succeeds, we now incorporate the interior constraints given by $X_1$ and $X_2$.
These constraints introduce additional directed edges in the enhanced graph $\H_R^{X_1,X_2}$,
which may create new predecessors for cubes in the lower cut.

The second step consists of a graph traversal starting from the lower cut:
we iteratively add to the lower part of the lower dicut all cubes that have a directed edge
leading to an already known cube of the lower part.

If during this process a cube belonging to the top part of the upper cut is added, then no dicut
separating $c_0$ and $c_1$ exists.
In this case, the instance \texttt{Tiling$(R,X_1,X_2)$} is not feasible.
Otherwise the growth of the bottom part of the lower dicut eventually ends when it remains no predecessor to add, providing a dicut of the graph  $\H_R^{X_1,X_2}$ and thus a  tiling solution of \texttt{Tiling$(R,X_1,X_2)$} .

\subparagraph{Complexity.}
The time complexity of the first step is that of Thurston’s algorithm.
The second step visits only cubes lying between the lower and upper dicuts of
$\H_R^{\emptyset,\emptyset}$.
The vertical distance between these two cuts is bounded by the length of the boundary,
namely $|\partial R|$.
Therefore, the number of cubes in this intermediate region is at most $O(|R|\,|\partial R|)$,
which bounds the running time of the second step.

Overall, the advancing surface algorithm runs in time $O(|R|\,|\partial R|)$.

In the particular case of the Calisson puzzle
\texttt{Tiling$(\varhexagon_n,\emptyset,X_2)$},
the lower and upper cuts are trivial: they correspond respectively to the empty
and full $n\times n\times n$ cube.
Since $|R|=\Theta(n^2)$ and $|\partial R|=\Theta(n)$,
the advancing surface algorithm runs in time $O(n^3)$.

\begin{figure}[ht]
  \begin{center}
    \includegraphics[width=0.95\textwidth]{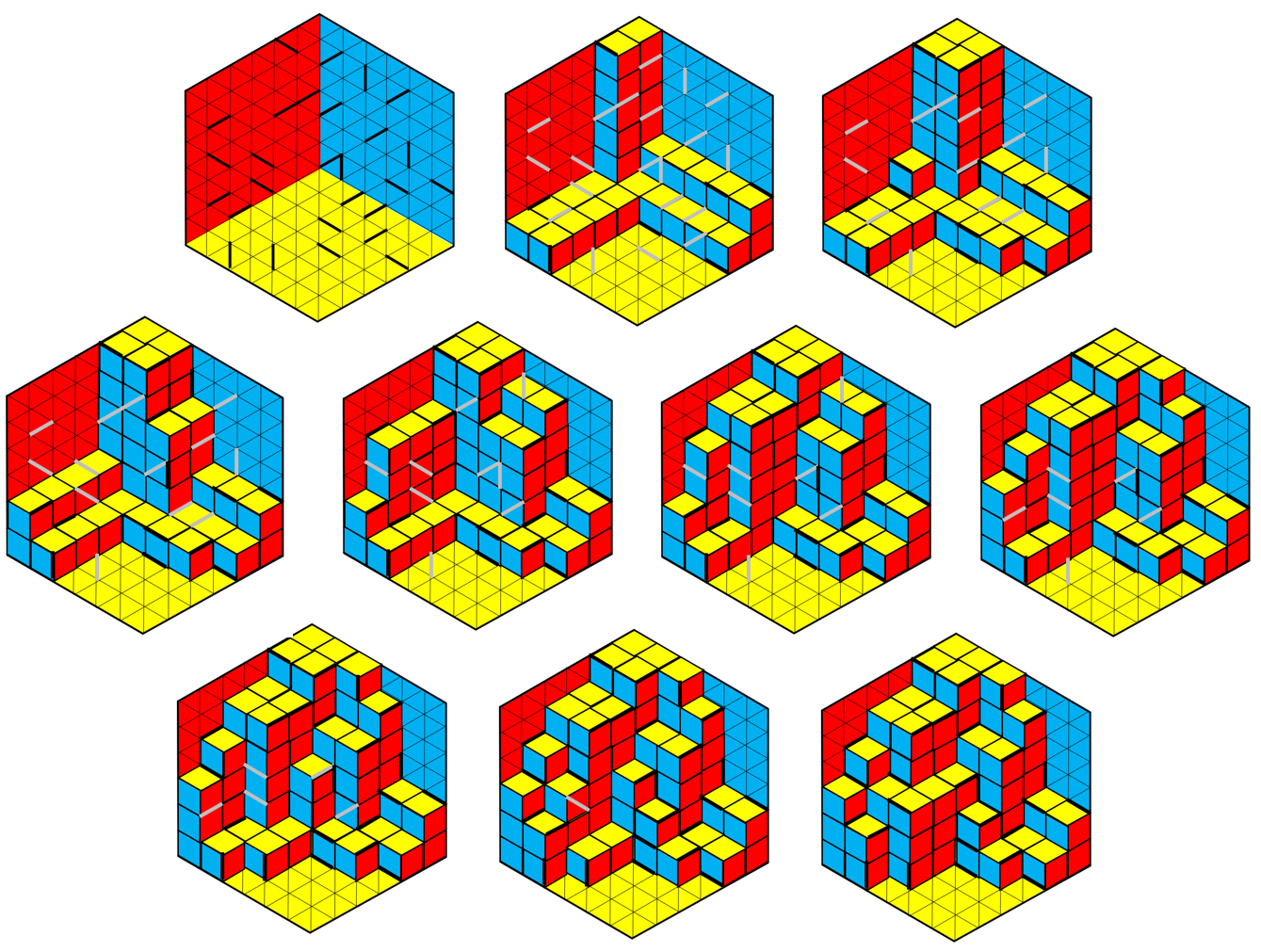}
  \end{center}
  \caption{\label{ads}
  \textbf{Solving an instance \texttt{Tiling$(\varhexagon_6,\emptyset,X_2)$} of the Calisson puzzle
  using the advancing surface algorithm.}
  All drawn interior edges belong to $X_2$.
  For a hexagonal region, the lower and upper cuts are trivial.
  At each step, the algorithm adds the minimal number of cubes needed to satisfy a violated constraint.
  Several additions may occur in parallel.
}
\end{figure}

\subparagraph{With a Pencil and a Rubber.}
Although the mathematical framework underlying the advancing surface algorithm is nontrivial,
its implementation in the context of the Calisson puzzle is remarkably simple.
Using only a pencil, an eraser, and a good three-dimensional intuition,
one can simulate on the paper the successive additions of cubes.
The previous analysis guarantees that this naive-looking procedure never misses a solution when one exists.
Figure~\ref{ads} illustrates several steps of the algorithm applied to one of the instances shown in Fig.~\ref{exercise}.

\subsection{Deciding whether the Whole Triangular Grid Can Be Tiled}

We now consider tiling instances of the form
\texttt{Tiling$(\triangle^2,X_1,X_2)$},
where the region $\triangle^2$ is the entire triangular grid and
$X_1$ and $X_2$ are two finite sets of input edges.
The problem is to decide whether such an instance admits a tiling.
This problem is of particular interest since, to the best of our knowledge,
none of the classical algorithms from the lozenge tiling literature
can be applied in this setting.
In contrast, our graph-theoretical framework yields a remarkably simple solution.
The key idea is to reduce the infinite system of difference constraints induced by
the graph $DC(\triangle^2,X_1,X_2)$
to an equivalent system induced by a finite graph.

\subparagraph{Decomposition into Positive and Negative Graphs.}
We now describe this reduction directly on the infinite graph
$DC(\triangle^2,X_1,X_2)$.
The strategy consists in decomposing it into two directed weighted graphs,
called the \emph{positive graph} $G^+$ and the \emph{negative graph} $G^-$.

\begin{itemize}
    \item The positive graph $G^+$ contains all strictly positively weighted edges of
    $DC(\triangle^2,X_1,X_2)$ and their incident vertices.
    Its vertex set is $\triangle^0$, hence it is infinite.
    By construction, all its edges are directed along the three lattice directions
    and have weight $+1$.
    As a consequence, the graph $G^+$ has an infinite number of vertices but it is highly regular and its distance function can be computed explicitly (Fig.~\ref{distances}).

    \item The negative graph $G^-$ contains all edges of weight $0$ or $-1$,
    namely the edges encoding the non-overlapping constraints of $X_1 \cup X_2$
    and the lateral edges encoding the saliency constraints of $X_2$,
    together with their incident vertices.
    Since $X_1$ and $X_2$ are finite, the graph $G^-$ is finite.
\end{itemize}

We now construct a finite graph $G^{+-}$ as follows.
Starting from $G^-$, we add a directed edge from any vertex $u$ of $G^-$ to any other
vertex $v$ of $G^-$.
The weight of this new edge is defined as the distance from $u$ to $v$ in the positive graph $G^+$.

\begin{claim}\label{cl}
The system of difference constraints induced by the infinite graph
$DC(\triangle^2,X_1,X_2)$ is feasible if and only if the system of difference constraints
induced by the finite graph $G^{+-}$ is feasible.
\end{claim}

\begin{proof}
We first prove the forward implication.
Assume that the system of difference constraints induced by $G^{+-}$ is not feasible.
Then $G^{+-}$ contains a directed cycle of strictly negative total weight.
By construction of $G^{+-}$, each positively weighted edge of this cycle
corresponds to a shortest path in the positive graph $G^+$.
Replacing each such edge by the corresponding path yields a directed cycle
in $DC(\triangle^2,X_1,X_2)$ with the same total weight.
Hence, $DC(\triangle^2,X_1,X_2)$ contains a strictly negative cycle,
and its system of difference constraints is not feasible.

Conversely, assume that the infinite graph $DC(\triangle^2,X_1,X_2)$
contains a directed cycle of strictly negative total weight.
This cycle necessarily contains at least one edge of non-positive weight,
and thus intersects the negative graph $G^-$.
Any maximal subpath of this cycle consisting solely of positively weighted edges
can be replaced by a single edge of $G^{+-}$ whose weight is the corresponding
shortest-path distance in $G^+$.
This replacement does not increase the total weight of the cycle.
As a result, we obtain a directed cycle of strictly negative total weight in $G^{+-}$,
showing that the system of difference constraints induced by $G^{+-}$ is not feasible.
\end{proof}

\begin{figure}[ht]
  \begin{center}
    \includegraphics[width=0.35\textwidth]{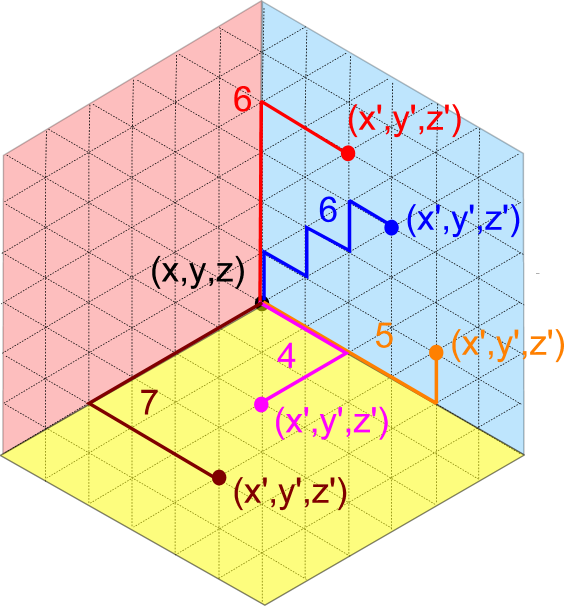}
  \end{center}
  \caption{\label{distances}
  \textbf{Distances in the positive graph $G^+$ on the triangular grid.}
  For vertices $\varphi(x,y,z)$ and $\varphi(x',y',z')$,
  the distance is $
  d = (x'-x)+(y'-y)+(z'-z)
  - 3 \min\{x'-x,\,y'-y,\,z'-z\}.
  $
  }
\end{figure}

The claim~\ref{cl} allows us to reduce the infinite system of difference constraints
associated with $DC(\triangle^2,X_1,X_2)$ to the finite system induced by $G^{+-}$.
The resulting algorithm for solving
\texttt{Tiling$(\triangle^2,X_1,X_2)$}
consists in constructing $G^{+-}$ and running the Bellman-Ford algorithm on it.
The distances in $G^+$ can be computed in constant time using the formula
$$d(\varphi(x,y,z) \rightarrow \varphi(x',y',z')) = (x'-x)+(y'-y)+(z'-z)
  - 3 \min\{x'-x,\,y'-y,\,z'-z\}
  $$

illustrated in Fig.~\ref{distances}.
If $n$ denotes the total number of edges in $X_1 \cup X_2$,
then the number of vertices of $G^{+-}$ is $O(n)$,
and the number of edges is $O(n^2)$.
Consequently, Bellman--Ford runs in time $O(n^3)$.
This proves Corollary~\ref{col}.

This final algorithmic result illustrates once again that the graph-theoretical
and difference-constraints layer added to Thurston’s theory
is not merely a reformulation,
but a powerful extension that makes it possible to solve tiling problems
far beyond the scope of classical methods.

\subsection*{Future Works}

A natural direction for future work is to investigate whether the graph
and difference constraints overlay introduced in this paper
can be extended to other tiling problems that admit a height-function formulation.
Is this approach specific to lozenge tilings or does it apply to a broader class of planar tilings whose configurations
can be encoded by monotone surfaces?

As an illustration, Fig.~\ref{execut} shows how domino tilings with interior
non-overlapping constraints can also be solved by reducing the problem
to a system of difference constraints and computing shortest paths.
This example strongly suggests that the graph-theoretical framework
developed in the previous pages may provide a unifying algorithmic viewpoint
for a wide range of constrained tiling problems.

\begin{figure}[ht!]
  \begin{center}
    \includegraphics[width=\textwidth]{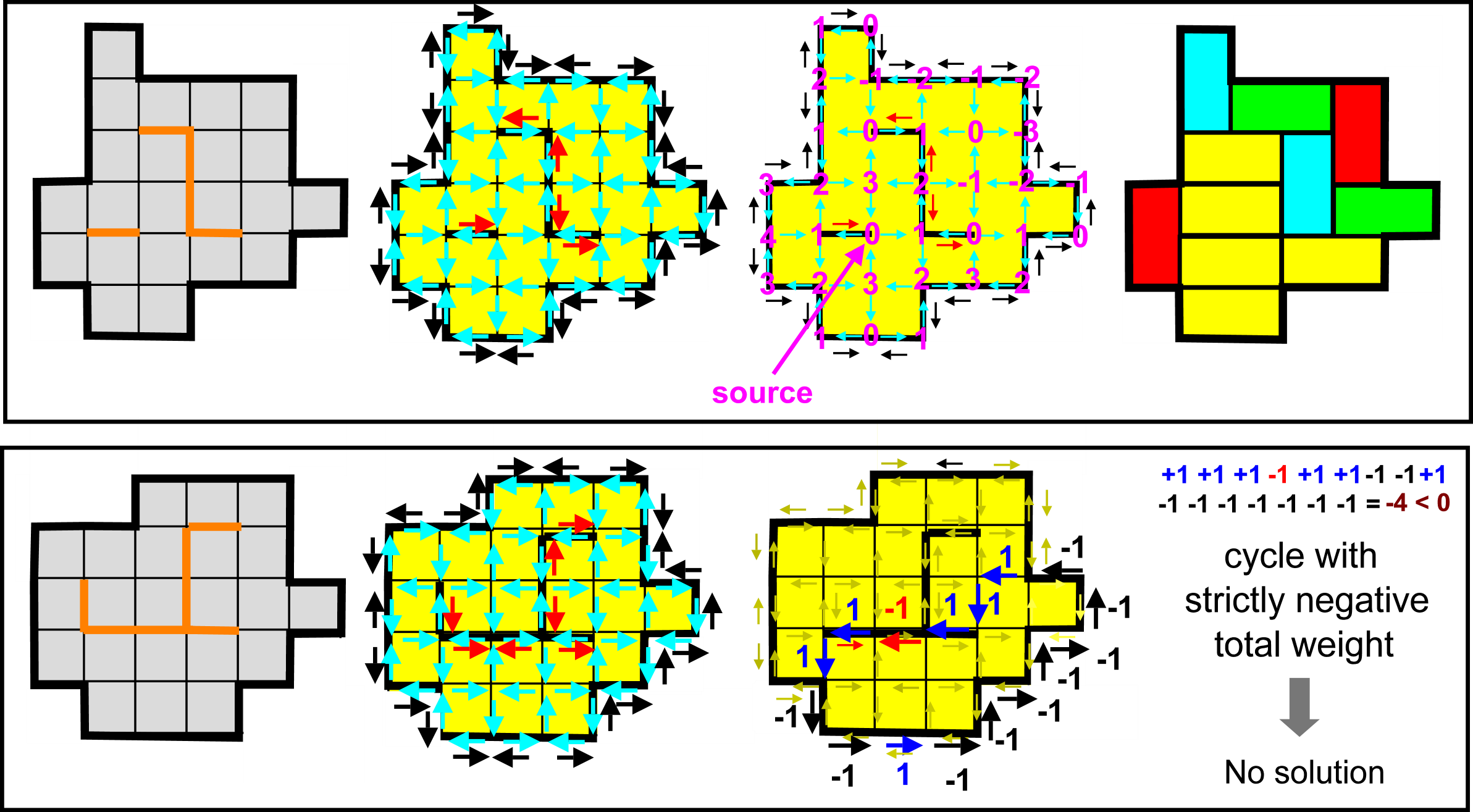}
  \end{center}
  \caption{\label{execut}
  \textbf{Domino tilings computed via shortest paths.}
  On the left, two instances of domino tilings with interior edges (in orange)
  that must not be overlapped.
  The corresponding weighted directed graph encoding the difference constraints
  contains blue, black, and red edges.
  Blue edges have weight $+1$ and are oriented clockwise around alternating faces
  (as on a chessboard).
  Black and red edges are oriented in the opposite direction and have weight $-1$.
  In the third column, distances are computed from an arbitrary source.
  As for lozenge tilings, if the graph contains a cycle of strictly negative total
  weight (second row), the tiling instance has no solution.
  Otherwise (first row), a tiling solution is obtained by connecting adjacent vertices
  whose distances to the source differ by $1$.
  }
\end{figure}

\bibliography{ref}
\end{document}